%% file: main.tex
\documentclass[11pt]{article}

\usepackage{amssymb,amsmath,amsthm,amsfonts}

\usepackage[T1]{fontenc}
\usepackage[tt=false, type1=true]{libertine}
\usepackage[varqu]{zi4}
\usepackage[libertine]{newtxmath}

\usepackage[margin=1in]{geometry}
\usepackage{mathtools}
\usepackage[bb=dsserif]{mathalpha}
\usepackage{graphicx}
\usepackage{booktabs}
\usepackage{algorithmic}
\usepackage[switch]{lineno}
\graphicspath{{Figures/}}
\usepackage{amsthm,bm}
\usepackage{pbox}
\usepackage{mathrsfs}
\usepackage{enumitem}
\usepackage{subcaption}

\newtheorem{theorem}{Theorem}

\newtheorem{lemma}{Lemma}

\newtheorem{remark}{Remark}
\newtheorem{proposition}{Proposition}
\newtheorem{assumption}{Assumption}

\usepackage{tabularx, booktabs}
\usepackage{multirow}

\usepackage{pifont}
\newcommand{\cmark}{\ding{51}}%
\newcommand{\xmark}{\ding{55}}%

\usepackage{hyperref}
\hypersetup{
    colorlinks,
    allcolors=myblue
}

\usepackage[sortcites,sorting=nyt,style=alphabetic]{biblatex}
\input{commands}

\theoremstyle{definition}

\newtheorem{definition}{Definition}

\newcommand{\citet}[1]{\textcite{#1}}

\title{Knowing When to Stop Matters: A Unified Framework for\\ Online Conversion under Horizon Uncertainty}

\author{
    Yanzhao Wang\thanks{University of Alberta. Email: \texttt{yanzhao5@ualberta.ca}}\\
    \and
    Hasti Nourmohammadi Sigaroudi\thanks{University of Alberta. Email: \texttt{hnourmoh@ualberta.ca}}\\
    \and
    Bo Sun\thanks{University of Waterloo. Email:
    \texttt{bo.sun@uwaterloo.ca}}\\
    \and 
    Omid Ardakanian\thanks{University of Alberta. 
    Email: \texttt{ardakanian@ualberta.ca}}\\
    \and
    Xiaoqi Tan\thanks{University of Alberta. 
    Email: \texttt{xiaoqi.tan@ualberta.ca}}
}

\date{\vspace{-25pt}}

\begin{document}

\maketitle

\begin{abstract}
    This paper investigates the online conversion problem, which involves sequentially trading a divisible resource (e.g., energy) under dynamically changing prices to maximize profit. A key challenge in online conversion is managing decisions under horizon uncertainty, where the duration of trading is either known, revealed partway, or entirely unknown. We propose a unified algorithm that achieves optimal competitive guarantees across these horizon models, accounting for practical constraints such as box constraints, which limit the maximum allowable trade per step. Additionally, we extend the algorithm to a learning-augmented version, leveraging horizon predictions to adaptively balance performance: achieving near-optimal results when predictions are accurate while maintaining strong guarantees when predictions are unreliable. These results advance the understanding of online conversion under various degrees of horizon uncertainty and provide more practical strategies to address real world constraints.
\end{abstract}

\section{Introduction}

The online conversion (\OC) problem models a sequential decision-making process in which a player aims to sell a fixed amount of a resource at time-varying prices~\cite{time_series_search_2001,k_search_2009,Sun_MultipleKnapsack_2020}. At each step, the player observes the current selling price and must irrevocably decide how much of the resource to trade, without knowledge of future prices. The goal is to maximize the total profit by the end of the time horizon.
This \OC problem captures the core decision-making process of numerous real-world applications, ranging from financial trading \cite{time_series_search_2001,k_search_2009} to online selection \cite{jiang2021online} and energy management \cite{Sun_MultipleKnapsack_2020,lee2024online,lin2019competitive}.

The \OC problem primarily involves two types of uncertainties: time-varying prices and the time horizon.
Recent work have focused on studying this problem under the assumption that uncertain prices are bounded within a finite support, while the time horizon is completely unknown. Although this classical setting has been extensively studied, there remains a gap between this theoretical uncertainty model and practical scenarios.
In many cases, the player (partially) knows information about the time horizon. For example, the number of steps in financial trading or energy trading is often fixed and known beforehand.

In the literature, for a given time horizon, the best-known result is a threat-based algorithm~\cite{time_series_search_2001}, which has been shown to be optimal under competitive analysis.
However, this algorithm is disconnected from the more recent and popular threshold-based algorithms~\cite{Sun_MultipleKnapsack_2020,lin2019competitive} for \OC problems. Furthermore, it remains unclear how to address horizon uncertainty when the player receives only last-minute notice of the horizon or is provided with a prediction with some degree of accuracy. Motivated by this research gap, the key question we aim to address in this paper is:
\begin{center}
How can we design online algorithms for the \OC problem to account for\\ varying degrees of horizon uncertainty? 
\end{center}

To address the above question, in this paper, we consider four types of horizon uncertainty models: (i) \OCK, a known horizon upfront; (ii) \OCN, a notified horizon, where the player does not know the horizon initially but is informed of it midway through the decision-making process; (iii) \OCU, a completely unknown horizon; and (iv) \OCP, an (imperfectly) predicted horizon, potentially obtained using AI or machine learning tools.
We first propose a unified algorithm for the \OC problem under the first three uncertainty models and prove that this algorithm achieves tight guarantees under competitive analysis.
Furthermore, to leverage horizon predictions, we adopt the emerging learning-augmented algorithm framework~\cite{mitzenmacher2022algorithms,lykouris2021competitive} and design an algorithm that robustly incorporates imperfect predictions. In particular, our main contributions are summarized as follows.

\begin{itemize}
    \item  We develop an online algorithm that unifies existing optimal results for three horizon settings, including the known, notified, and unknown horizons. We show that the celebrated threat-based algorithm from \cite{time_series_search_2001} can be unified within our framework, offering a new interpretation of its mechanism through the lens of pseudo-revenue maximization.
    
    \item  Compared to the classic \OC problem, our algorithm can additionally account for box constraints, i.e., the rate limit for trading in each step. We establish a tight upper bound for the case with a known horizon, providing, to the best of our knowledge, the first algorithm with tight guarantees for \OC under fixed horizons and non-trivial rate constraints.
    
    \item  We extend our unified algorithm by incorporating predictions about the time horizon. This learning-augmented algorithm achieves the best-known performance guarantees in terms of \textit{robustness} (maintaining strong worst-case performance when predictions are highly inaccurate) and \textit{consistency} (achieving optimal performance when predictions are accurate). We further show that even when the prediction underestimates or exactly matches the actual horizon, the algorithm outperforms those without access to horizon information, demonstrating its practical utility in scenarios with partial predictions.
\end{itemize}

Our contributions, particularly the handling of horizon uncertainty and box constraints, represent a notable advancement in \OC algorithm design and analysis. With a known horizon and non-trivial box constraints, precise scheduling of resource allocation at each time step becomes highly correlated. To address this challenge, we introduce a two-phase algorithmic design that partitions trading decisions into distinct phases. This structure reduces the problem to a simpler form without loss of generality, enabling the derivation of tight results using an online primal-dual framework.

\begin{table*}
\footnotesize
    \centering
    \def\arraystretch{1.1}
    \begin{tabular}{p{1.4cm} c c c c c c c c}
        \toprule
           & \multicolumn{2}{c}{\bf \OCK} & \multicolumn{2}{c}{\bf \OCN} & \multicolumn{2}{c}{\bf \OCU} & \multicolumn{2}{c}{\bf \OCP} \\ 
         \cmidrule(lr){2-3}  \cmidrule(lr){4-5} \cmidrule(lr){6-7}  \cmidrule(lr){8-9}
          \bf Papers  & w/ box  & w/o box & w/ box &  w/o box & w/ box & w/o box  & w/ box & w/o box \\   
         \midrule
         \cite{time_series_search_2001} & \xmark &  \cmark & \xmark & \cmark & \xmark & \xmark & \xmark & \xmark \\\hline
         \cite{k_search_2009} & \xmark & \xmark & \cmark & -- & \xmark & \xmark & \xmark & \xmark\\\hline
         \cite{tan2023threshold}
         & \xmark & \xmark & \xmark & \xmark & \cmark & -- & \xmark & \xmark \\\hline
         \cite{lin2019competitive}  & \xmark & \xmark & \xmark & \xmark & \xmark & \cmark & \xmark & \xmark \\\hline
         \cite{lechowicz2024online} & \xmark & \xmark & \cmark & \cmark & \xmark & \xmark & \xmark & \xmark \\\hline
         \cite{Sun_MultipleKnapsack_2020} & \xmark & \xmark & \xmark & \xmark & \cmark & \cmark & \xmark & \xmark\\  \hline
         This paper & \cmark & \cmark & \cmark & \cmark & \cmark & \cmark & \cmark & \cmark \\ 
         & $\textbf{Tight}^\star$ & $ \text{Optimal}^\star$ & $ \text{Optimal}^\star$ & $ \text{Optimal}^\star$ & Optimal & Optimal & $ \textbf{Best-known}^\star$ & $ \textbf{Best-known}^\star$ \\
         \bottomrule
    \end{tabular}
    \bigskip
    \caption{Comparison of our work with the existing literature on online conversion (\OC) problems. $ \textbf{Bold}^\star$ indicates new results obtained by this paper, {\normalfont{Optimal}} denotes state-of-the-art optimal results that can be unified by our algorithm, and $ \text{\normalfont{Optimal}}^\star $ represents optimal results obtained both by existing studies and this paper, but through different approaches. \OCK, \OCN, and \OCU denote the \OC problem with known horizon, notified horizon partway through, and unknown hrizon, respectively. \OCP is the learning-augmented setting with machine learned horizon predictions.}
    \label{tab:literature}
\end{table*}

\subsection{Related Work}
 
\paragraph{Online conversion under horizon uncertainty} 
The \OC problem has been studied under different horizon settings. In the \OCK case, the total number of decision steps is fixed, allowing structured strategies to achieve optimal competitive ratios~\cite{time_series_search_2001,k_search_2009}. The \OCN setting, where the horizon is revealed partway through, enables partial adjustments to improve performance~\cite{lechowicz2024online,Sun_MultipleKnapsack_2020}. In the \OCU case, algorithms focus on robustness to ensure worst-case guarantees~\cite{tan2023threshold,lin2019competitive}. While these works address horizon uncertainty, the integration of box constraints, which limit resource allocation at each step, remains largely unexplored. In particular, the case of a \OCK with box constraints has not been studied, and a unified framework for all settings remain open challenges, which this work aims to resolve.

\paragraph{Learning-augmented algorithms}
The emerging algorithmic framework aims to enhance online decision-making by incorporating machine-learned predictions, bridging the gap between worst-case guarantees and average-case performance \cite{lykouris2021competitive,purohit2018improving}. This framework has been applied to problems like caching \cite{lykouris2021competitive}, ski-rental \cite{purohit2018improving}, and bin packing \cite{angelopoulos2023online}, demonstrating improved outcomes when predictions are accurate while maintaining robustness against prediction errors. In the context of \OC, prior work has developed learning-augmented algorithms that leverage predictions to improve online decision-making. Recent studies~\cite{angelopoulos2022online,angelopoulos2024overcoming,sun2021pareto} have specifically focused on incorporating price predictions to enhance adaptability in \OC problems. While these contributions highlight the potential of prediction-based approaches, they do not address the use of horizon predictions. This paper fills this gap by introducing learning-augmented algorithms that leverage horizon predictions, enabling more effective planning of resource allocation while ensuring reliable performance even when predictions are imperfect.

Table~\ref{tab:literature} summarizes existing work on \OC problems and highlights our contributions. For the known horizon case (\OCK), previous work~\cite{time_series_search_2001} does not account for box constraints, whereas our work addresses this limitation. Moreover, our approach differs from the threat-based algorithm in~\cite{time_series_search_2001} and the threshold-based algorithms in~\cite{Sun_MultipleKnapsack_2020,k_search_2009}, offering a novel perspective on \OCK. For the notified horizon setting (\OCN), related work~\cite{k_search_2009,lechowicz2024online} primarily focuses on binary trading decisions, while our algorithm handles continuous resource allocation under rate constraints. In the unknown horizon setting (\OCU), work by~\cite{Sun_MultipleKnapsack_2020,lin2019competitive,tan2023threshold} addresses both constrained and unconstrained cases, which we successfully unify within our algorithmic framework. For the prediction setting (\OCP), no prior work has investigated horizon predictions for \OC problems. Our work addresses this unexplored area by integrating horizon predictions into the algorithm, achieving the best-known results in terms of performance under uncertain conditions.

\section{Problem Statement}
In this section, we formulate the \OC problem with box constraints. Based on the available information about the horizon, we further define four versions of \OC problems.

\subsection{Online Conversion Problems}

In an \OC problem, a player aims to allocate $k$ units of a divisible resource over a series of time steps, with prices revealed sequentially. At each time step $t \in [T]$, the player observes the price $p_t$ and makes an irreversible decision on how much resource $x_t \in [0, b]$ to allocate, where $b$ is the rate limit, i.e., the maximum units of resource that can be allocated in each step. The goal is to maximize the total profit over the $T$ steps, while satisfying the total resource constraint.

In the offline setting, where the price information $\sigma = \{p_1, \ldots, p_T\}$ is known, the optimal profit is denoted by $ {\OPT}(\sigma) $, and can be solved by the following linear program:
\begin{subequations}\label{eq:OC_max}
\begin{align} 
    \max_{ \{x_t\}_{t} } \quad &\sum\nolimits_{t=1}^T p_t x_t\\
    {\rm s.t.} \quad & \sum\nolimits_{t=1}^T x_t \leq k,                  \label{eq:OC_max_budget_constraint}  \\ 
     & 0 \leq  x_t \le b,  \quad  \forall t \in [T].   \label{eq:OC_max_box_constraint}    
\end{align}
\end{subequations}
We aim to design online algorithms that minimize the worst-case competitive ratio (\CR):
\begin{align}\label{cr_max}
    \CR = \max_{\sigma \in \Omega} \frac{{\OPT}(\sigma)}{{\ALG}(\sigma)},
\end{align}
where $ {\ALG}(\sigma) $ denotes the total profit achieved by the online algorithm for given $ \sigma $, and $ \Omega $ represents the family of all instances under a particular uncertainty model.

The \OC problem is a classic online problem for sequential decision-making in volatile markets~\cite{time_series_search_2001,k_search_2009,yang2020online,sun2021pareto}. In this paper, we focus on two aspects of \OC problems: (i) \textit{horizon uncertainty}, namely, the uncertainty of $T$, and (ii) \textit{box constraints} (i.e., the rate constraints \eqref{eq:OC_max_box_constraint} at each step). In the following, we say that the box constraint \eqref{eq:OC_max_box_constraint} is \textit{\bfseries non-trivial} if $b \in (k/T, k)$, and \textit{\bfseries trivial} if $b \in (0, k/T]$ or $b \in [k, +\infty)$. For example, for $b \in (0, k/T]$, greedily allocating as much as possible at each step (i.e., setting $x_t = b$ for all $t \in [T]$) is the only optimal strategy. In the case where $b \in [k, +\infty)$, the box constraints are never binding and can therefore be removed.

\subsection{Horizon Uncertainty Models}
One of the key assumptions in our model is that the price at each time step, denoted as $p_t$, is bounded by a known interval:

\begin{assumption}
The price $p_t$ at each step $t$ is bounded within a known range:
\begin{align*}
p_t \in [p_{\min}, p_{\max}], \quad \forall t \in [T].
\end{align*}
\end{assumption}

Additionally, we define $\theta = \frac{p_{\max}}{p_{\min}}$ to represent price fluctuation. When $\theta = 1$, there is no price variation, simplifying the problem to $ \CR = 1 $. As $\theta$ increases, however, the problem becomes more complex due to larger price fluctuations.

We study the following four horizon uncertainty models:
\begin{itemize}
    \item \textbf{\OCK}: The player knows the total number of time steps $T$ in advance. This setting is represented by the family of instances: $ \Omega_{\textsf{known}}(p_{\min}, p_{\max}, T)$. 

    \item \textbf{\OCN}: The player does not know the value of $T$ a priori but receives a notification at a specific step if all future allocations must be performed at the maximum rate to use the entire resource by the end. For example, consider \OC with $k = 12$ and $b = 3$. The player initially does not know the value of $T$ and decides to allocate $x_t = 1$ for $t = 1, 2, 3$. At the beginning of step $t=4$, the player is notified that $T = 6$, requiring $x_4 = x_5 = x_6 = 3$ to ensure no resource is left unused. The family of instances for this setting is denoted by $ \Omega_{\textsf{notice}}(p_{\min}, p_{\max}) $. 

    \item \textbf{\OCU}: The player has no information about the total number of time steps $T$ throughout the allocation process (i.e., the process may end suddenly). The family of instances for this setting is denoted by $ \Omega_{\textsf{unknown}}(p_{\min}, p_{\max}) $.

    \item \textbf{\OCP}: The player has access to a predicted horizon $T_{\textsf{pred}}$, which may differ from the actual horizon $T$. The predicted horizon introduces additional uncertainty, as decisions are made based on an imperfect forecast of $T$. This setting is represented by the family of instances $ \Omega_{\textsf{prediction}}(p_{\min}, p_{\max}, T_{\textsf{pred}})$.  If the player ignores the prediction, \OCP reduces to \OCU. Conversely, if the player fully trusts an accurate prediction where $T_{\textsf{pred}} = T$, \OCP reduces to \OCK.
\end{itemize}

Each of these settings introduces varying degrees of uncertainty regarding the horizon $T$, with higher uncertainty increasing the difficulty of optimizing allocation decisions.

\section{A Unified Algorithm for \OCK, \OCN, and \OCU}
This section presents a unified algorithm for the \OC problem under \OCK, \OCN, and \OCU. We start by introducing pseudo-cost functions,  a core concept of our unified algorithm.

\subsection{Pseudo-Cost Functions}\label{sec:def_pseudo_cost_functions}
Consider \OC in the offline setting, where the complete price sequence $ \sigma = \{p_1, \ldots, p_T\} $ is known in advance. The dual of the maximization problem \eqref{eq:OC_max} can be expressed as:
\begin{subequations}\label{dual_OC_max}
\begin{align}
   \min_{ \varphi, \bm{\mu}}\quad & k\varphi+\sum_{t=1}^{T} b\mu_t \\
    s.t. \quad & \varphi + \mu_t \ge p_t, & & \forall t\in[T],\\ 
                & \varphi, \mu_t \ge 0,  & & \forall t\in[T],
\end{align}
\end{subequations}
where $\varphi$ represents the marginal value of the resource, and $\mu_t$ adjusts for the box constraint $x_t \leq b$. These dual variables balance the value of resource allocation while respecting stepwise limits. In the online setting, as prices are revealed sequentially, we estimate the marginal value dynamically, adapting to observed prices and constraints using \textit{pseudo-cost functions}, denoted by $ \boldsymbol{\phi} = \{\phi_t\}_{\forall t \in [T]} $, which approximate the marginal value of the resource.

\begin{definition}[{\bfseries Pseudo-Cost Functions}]\label{def_price}
Given $ \alpha \geq 1 $, the pseudo-cost function $ \boldsymbol{\phi} = \{\phi_t\}_{\forall t \in [T]} $ is defined as:
\begin{align}\label{eq_def_phi_t}
   \phi_t(\beta|  {F}_t,\alpha) = p_{\min} + \frac{(\alpha - 1)p_{\min}}{(1 - \alpha \beta/k) \prod_{i=1}^{t-1}(1 - \alpha x_i/k)}.
\end{align}
${F}_t = (\mathbf{x}_{t-1}, \sigma_t) $ denotes all historical information up to step $ t $, where $ \mathbf{x}_{t-1} = \{x_1, \ldots, x_{t-1}\} $ represents decisions up to $ t-1 $ and $ \sigma_t = \{p_1, \ldots, p_t\} $ denotes prices observed up to $ t $.
\end{definition}

At first glance, the pseudo-cost function defined by Eq. \eqref{eq_def_phi_t} may seem arbitrary. To clarify its purpose and trade-offs, we provide the following remarks:

\begin{itemize}
    \item The pseudo-cost $ \phi_t $ at time $ t $ is an \textit{$ \alpha $-parameterized} function, monotonically increasing in $ \alpha \in [1, +\infty)$. Here, $ \alpha $ serves as a \textit{balance parameter}, controlling the level of the pseudo-cost. A larger $ \alpha $ increases $ \phi_t $, encouraging the player to trade resources at higher perceived prices. Thus, $ \alpha $ must be carefully tuned to ensure the perceived pseudo-cost is neither excessively low nor overly high.

    \item Unlike traditional threshold-based algorithms (e.g., \cite{Tan_MD_2020,lechowicz2024online,Sun_MultipleKnapsack_2020}), the pseudo-cost function in Eq. \eqref{eq_def_phi_t} is \textit{adaptive to all historical trading decisions} (i.e., $ \mathbf{x}_{t-1} = \{x_1, \ldots, x_{t-1}\} $). This fine-grained adaptation enables our approach to effectively address \OC problems with known horizons, achieving the same optimal \CR as threat-based algorithms \cite{time_series_search_2001}. Furthermore, the pseudo-cost function establishes a direct connection between the resource's perceived value and market prices, linking the primal and dual terms in Eqs. \eqref{eq:OC_max} and \eqref{dual_OC_max}, the important properties of Eq. \eqref{eq_def_phi_t} are detailed in Appendix \ref{sec:property_of_phi}. This connection underpins the use of the online primal-dual (OPD) framework \cite{Buchbinder_OPD_book_2009} for designing and analyzing the unified algorithm presented in Algorithm \ref{alg_RBP}.
\end{itemize}

\begin{algorithm}
\caption{Unified Algorithm for \OC via Pseudo-Revenue Maximization ($\PRM_{\bm{\phi}}$)}
\label{alg_RBP}
\begin{algorithmic}[1]
\STATE \textbf{Inputs:} Pseudo-cost functions $\phi = \{\phi_t\}_{\forall t \in [T]}$, balance parameter $\alpha$, initial resource $k$, time horizon $T$ ($T = \infty$ for \OCN and \OCU), and box-constraint $b$.

\WHILE{price $p_t$ and notification value $n_t$ are revealed}
    \STATE Solve the following optimization problem:
    \begin{align}
       x_t^* = \argmax_{x_t \in [0, b]} p_t x_t - \int_{0}^{x_t} \phi_t(\beta| F_t, \alpha) d\beta. \label{pseudo_revenue}
    \end{align}
    \IF{$\sum_{i=1}^t n_i > 0$} 
        \STATE Allocate $\bar{x}_t = \min\{b, k_{t-1} - b(T - t)\}$ 
        \label{line_forced_1}
    \ELSIF{$k_{t-1} - x_t^* \leq b(T - t)$} 
    \label{line_query}
        \STATE Allocate $\bar{x}_t = \min\{x_t^*, k_{t-1}\}$ 
        \label{line_proactive}
    \ELSE
        \STATE Allocate $\bar{x}_t = \min\{b, k_{t-1} - b(T - t)\}$ 
        \label{line_forced_2}
    \ENDIF
    \STATE Update remaining resource: $k_t = k_{t-1} - \bar{x}_t$.
\ENDWHILE
\end{algorithmic}
\end{algorithm}

\subsection{How \PRM Works}
Algorithm \ref{alg_RBP} dynamically optimizes resource allocation based on the current price $p_t$ and pseudo-cost $ \phi_t $. It adapts to different horizon settings:
\begin{itemize}
\item \textbf{\OCK:} The total horizon $T$ is known. The algorithm checks (line \ref{line_query}) if allocating $x_t^*$ can fully utilize $k$ over the remaining steps. If feasible, it stays proactive (line \ref{line_proactive}); otherwise, it switches to forced allocations (line \ref{line_forced_2}), trading $ b $ until all remaining resources are allocated.
\item \textbf{\OCN:} The horizon is unknown, but a notification $n_t$ (indicating proximity to the end) may arrive. Upon notification (denoted by $n_t = 1$), the algorithm transitions to the forced phase (line \ref{line_forced_1}).
\item \textbf{\OCU:} The horizon is entirely unknown. The algorithm remains in the proactive phase, assuming infinite opportunities.
\end{itemize}

The balance parameter $\alpha$ plays a crucial role in determining the trade-off between immediate and future allocations. When $\alpha = 1$, the strategy is aggressive, allocating at the lowest possible price $p_{\min}$. Larger $\alpha$ values lead to more reserved strategies, holding resources for potentially higher future prices. Based on Algorithm \ref{alg_RBP}, for any given pseudo-cost function $ \phi $ that follows Definition \ref{def_price} with some balance parameter $ \alpha \geq  1 $, there exist the worst-case competitive ratio $ \CR(\alpha) $. Our goal is to design $ \alpha $ to minimize $ \CR(\alpha) $:
\begin{align}\label{cr^*_known}
\CR(\alpha) := 
\max_{\sigma \in \Omega} \frac{\OPT(\sigma)}{\PRM_{\boldsymbol{\phi}}(\sigma|\alpha)}.
\end{align}

For \OCK and \OCN, the competitive ratios are denoted by $ \CR_{\textsf{known}}^* $ and $ \CR_{\textsf{notice}}^* $, respectively.

\subsection{Competitive Ratios for \OC}
The main results for \OCK and \OCN are summarized in Theorem \ref{thm:oc_known_notice} below, which shows the optimal design of the balance parameter $ \alpha $ for the pseudo-cost function and the corresponding \CR in different horizon settings.

\begin{theorem}\label{thm:oc_known_notice}
If Algorithm \ref{alg_RBP} is executed with pseudo-cost function $ \phi $ as defined in Definition \ref{def_price}, then the following holds: 
\begin{itemize}
    \item For \OCK with non-trivial box constraints, $\PRM_{\boldsymbol{\phi}}$ is $ \CR^*_{\textsf{known}} $-competitive if we set $ \alpha = \CR^*_{\textsf{known}} $ as the root to the following equation:
    \begin{align}\label{eq:CR_known}
    \CR^*_{\textsf{known}} = \Big(T - \lceil \frac{k}{b} \rceil + 1\Big) \left[ 1 - \Big(\frac{\CR^*_{\textsf{known}} - 1}{\theta - 1} \Big)^{\frac{1}{T - \lceil \frac{k}{b} \rceil + 1}} \right].
    \end{align}
    \item For \OCN, $\PRM_{\boldsymbol{\phi}}$ is $ \CR^*_{\textsf{notice}} $-competitive if we set $ \alpha = \CR^*_{\textsf{notice}} $ as follows:
    \begin{align}\label{eq:CR_notice}
    \CR^*_{\textsf{notice}} = 1 + W\left(\frac{\theta - 1}{e}\right),
    \end{align}
    where $ W $ denotes the Lambert-$W$ function and $ \theta = p_{\max}/p_{\min} $ denotes the price fluctuation ratio.
\end{itemize}
\end{theorem}

The complete proof of Theorem \ref{thm:oc_known_notice} is given in Appendix \ref{sec:proof_of_thm_known_notice}, \ref{sec:proof_of_theorem_CR_box_constraint} and \ref{sec:proof_of_theorem_CR_box_constraint_notice}.  We give two remarks on Theorem \ref{thm:oc_known_notice}.

\begin{remark}\label{rmk:CR_unknown}
For \OCU, based on  \cite{Sun_MultipleKnapsack_2020}, we can show that $ \PRM_{\boldsymbol{\phi}} $ achieves the optimal \CR with a different pseudo-cost function. For ease of reference, the results from \cite{Sun_MultipleKnapsack_2020} are reviewed in Appendix \ref{appendix_OC_unknown} and below we give the optimal \CR, denoted by $ \CR^*_{\textsf{unknown}} $, for \OCU  without proof:
\begin{align}\label{eq:CR_unknown}
 \CR^*_{\textsf{unknown}} = 1 + \ln \theta.
\end{align}
Note that for both \OCN and \OCU, the box constraint $ b $ has no impact on the optimal CR, as indicated by Eqs. \eqref{eq:CR_notice} and \eqref{eq:CR_unknown}. 
\end{remark}

\begin{remark}
For \OCK without box constraints or with trivial box constraints (i.e.,  $  b \in [k, +\infty) $), Eq. \eqref{eq:CR_known} can be simplified as $ \CR^{*}_{\textsf{known}} = T [ 1 - (\frac{\CR^*_{\textsf{known}} - 1}{\theta - 1})^{1/T}]$. 
El-Yaniv et al. \cite{time_series_search_2001} has shown the same \CR, but through a different approach known as the threat-based algorithm. It is also worth noting that to prove the optimality of $ \CR^*_{\textsf{known}} $, we propose a new approach based on the Gronwall’s Inequality \cite{mitrinovic2012inequalities,jones1964fundamental}. More details are given in Appendix \ref{sec:lower_bound_proof}.
\end{remark}

\textbf{Theoretical implications of Theorem \ref{thm:oc_known_notice}}: Our results here have two important implications. First,  based on Theorem \ref{thm:oc_known_notice} and Remark \ref{rmk:CR_unknown}, $ \PRM_{\boldsymbol{\phi}} $ provides a unified algorithm for solving \OC under three different horizon settings, all through the lens of pseudo-revenue maximization. Second, to the best of our knowledge, $ \PRM_{\boldsymbol{\phi}} $ is the first to achieve a tight competitive ratio for \OCK under non-trivial box constraints.\footnote{Our competitive ratio $ \CR^*_{\textsf{known}} $ given in Eq. \eqref{eq:CR_known} is asymptotically optimal with respect to $ T $ and $ b $, meaning that, in the regime of large $ T $ (or large $ b $), $ \CR^*_{\textsf{known}} $ is optimal.} 

\begin{figure}[htbp]
    \centering
    \begin{subfigure}[b]{0.3\textwidth}
        \centering
        \includegraphics[width=\textwidth]{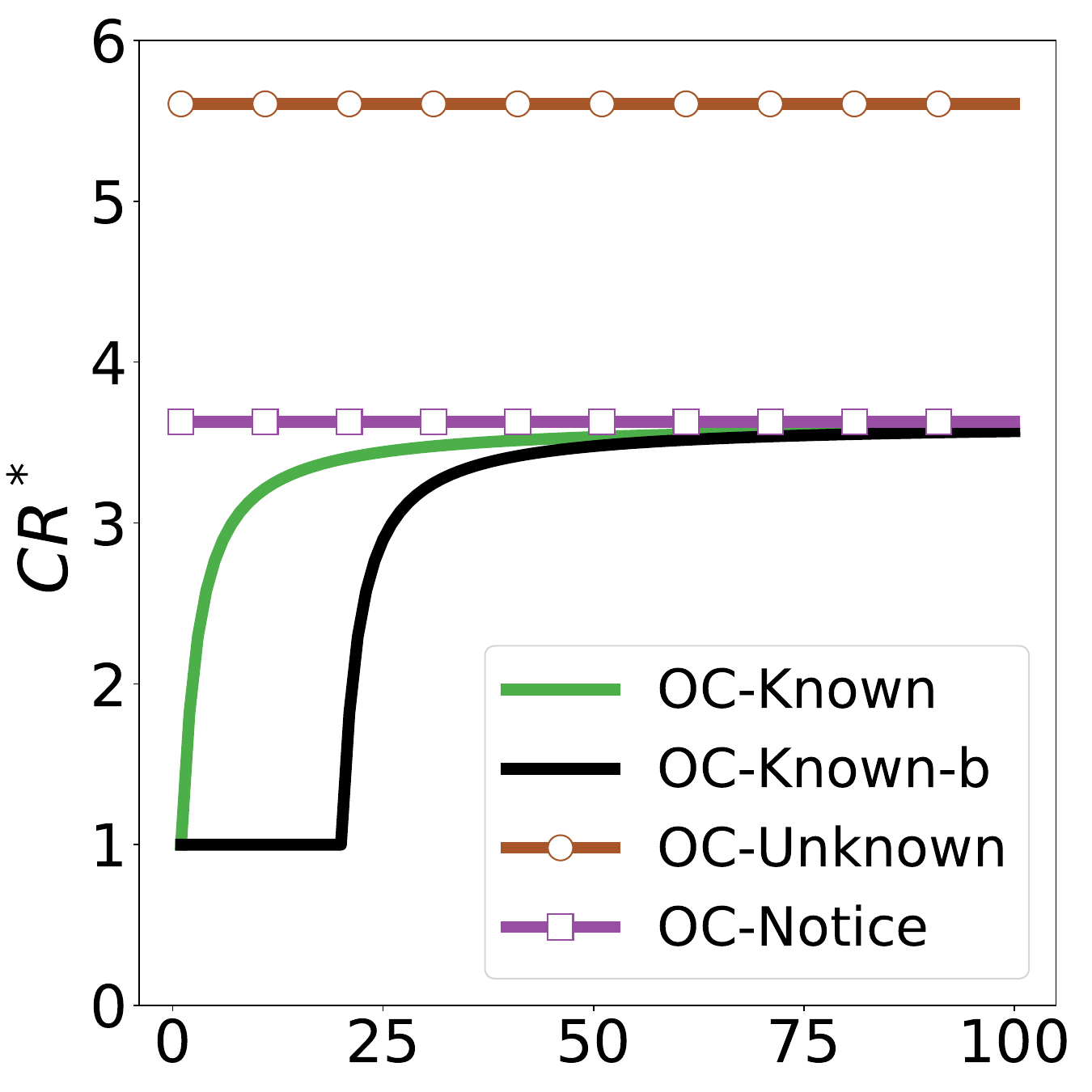}
        \caption{CR vs $T$}
        \label{fig:sub_cr_vs_T}
    \end{subfigure}
    \begin{subfigure}[b]{0.3\textwidth}
        \centering
        \includegraphics[width=\textwidth]{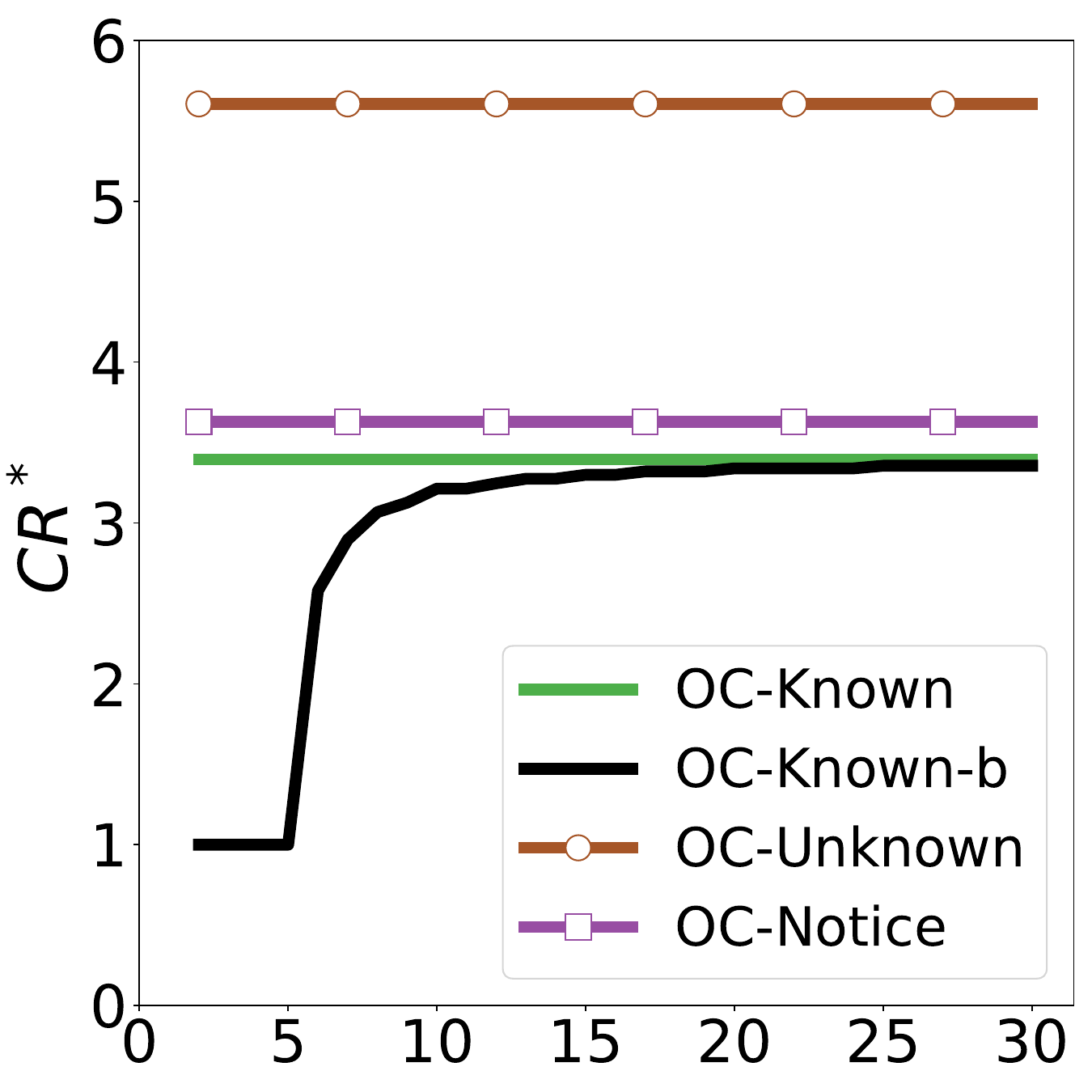}
        \caption{CR vs $b$}
        \label{fig:sub_cr_vs_b}
    \end{subfigure}
    \begin{subfigure}[b]{0.3\textwidth}
        \centering
        \includegraphics[width=\textwidth]{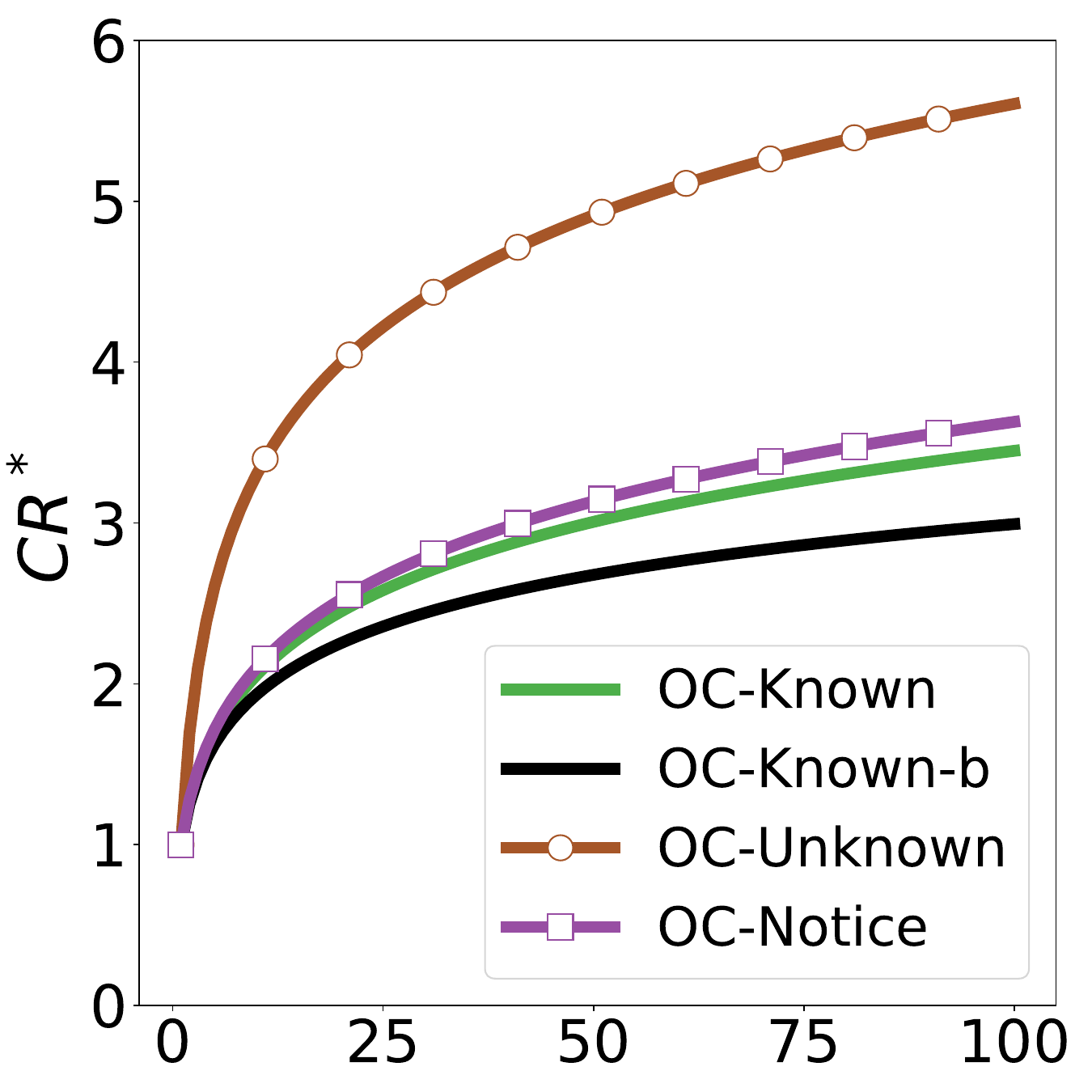}
        \caption{CR vs $\theta$}
        \label{fig:sub_cr_vs_theta}
    \end{subfigure}
    \caption{Comparison of \CR across different settings. \OCK refers to the setting without box constraints, while \OCK-b represents the same setting but with non-trivial box constraints.}
    \label{fig:cr_comparison}
\end{figure}

\textbf{Comparisons of optimal \CRs in different settings}: Fig. \ref{fig:cr_comparison} is presented to compare the optimal \CRs across four settings: \OCK with non-trivial box constraints (labelled as \OCK-$b$), \OCK without box constraints, \OCN, and \OCU. In Fig. \ref{fig:sub_cr_vs_T}, the \CRs are plotted as a function of increasing $T$, with $\theta = 100$ and $b = 5$. The gap between \OCK-$b$ and \OCK is noticeable at the beginning, but decreases as $T$ increases. \OCK without box constraints converges rapidly to \OCN, and \OCK-$b$ follows, eventually converging as well. \OCU shows significantly worse performance, emphasizing the importance of knowing $T$ in the online trading process. Fig. \ref{fig:sub_cr_vs_b}, with fixed $\theta = 100$ and $T = 20$, illustrates how \CR changes with respect to varying $b$. When $b$ is small, \CR equals $1$, as both \OPT and $\ALG_{\phi}$ will behave exactly the same (trade as fast as possible). As $b$ increases, \CR for \OCK-$b$ gradually converges to \OCK. Lastly, Fig. \ref{fig:sub_cr_vs_theta}, with fixed $T = 20$ and $b = 5$, shows that \CR consistently increases with $\theta$ across all the four settings. It represents a fact that larger $\theta$ reflects higher price fluctuations, which increases the difficulty for the seller in optimizing the trading process.

\section{Proof Sketch of Theorem \ref{thm:oc_known_notice} for \OCK with Non-trivial Box Constraints}
In this section, we provide a proof sketch to illustrate the key technical ingredients in deriving Theorem \ref{thm:oc_known_notice} for \OCK with non-trivial box constraints (i.e.., $ b\in (k/T, k)$).

\subsection{Step 1: Definitions and Notations of Switching Step and Two-Phase Trading}
For \OCK with non-trivial box constraints, both the optimal offline solution (\OPT) and our unified algorithm ($ \PRM_{\boldsymbol{\phi}}$) are guaranteed to utilize the entire budget $k$ within the time horizon $T$. The resource trading process is characterized by a critical transition, termed the \textit{switching step} $\tau$, which divides the trading into two phases: the \textit{proactive} phase and the \textit{forced} phase.

The switching step $\tau$ identifies when the algorithm transitions from proactive to forced trading due to constraints. To determine $\tau$, we use the concept of \textit{laxity}—a measure of how much flexibility remains in trading decisions relative to the remaining time and resource. Specifically, the laxity at step $t$, denoted as $l_t$, evaluates how freely the algorithm can continue trading before it is forced to allocate the maximum amount $b$.
Mathematically, we define $l_t$ at step $t\in[1,T]$ as follows: 
\begin{align} 
l_t = (T-t) - \lceil \frac{k_{t-1}-x_t}{b} \rceil, 
\end{align} 
where $k_{t-1}$ is the remaining resource before step $t$, $x_t$ is the resource allocated at step $t$, and $b$ is the rate limit per step. This measure reflects the remaining resource per time step relative to the box constraint imposed by $b$. When laxity drops below 0, the proactive trading strategy is no longer feasible, and the algorithm switches to a forced phase where it must allocate exactly $b$ units at each step to meet the overall budget by time $T$. In this regard, the switching step $\tau$ is the first step at which laxity becomes negative, signifying the transition to forced trading.

\begin{lemma}\label{lemma_switch}[\textbf{Switching Step} $ \tau $]
Let $ \tau $ be the first step such that  $ k_{\tau-1} - x_{\tau}^* >  b(T - \tau)  $ (i.e., $ l_t < 0$), where $ x_{\tau}^* $ denotes the optimal amount of resource allocated at step $ \tau $ by $ \PRM_{\boldsymbol{\phi}}$.  The allocated resource at each step is
\begin{align}\label{allocation_with_limit}
\bar{x}_t = 
\begin{cases}
    \min\{x_t^*, k_{t-1}\}                 &  t < \tau,\\
    \min\{b, k_{t-1} -  b(T - t) \}  &  t=\tau,\\
    b  &  \tau < t\le T.
\end{cases}
\end{align}
\end{lemma}
Lemma \ref{lemma_switch} describes the trading strategy at each step once the switching step $ \tau $ is given. Initially, the algorithm follows a proactive strategy, trading the optimal amount $ x_t^* $ while respecting the remaining resource and the box constraint $ b $. However, when laxity drops below 0, the algorithm switches to a forced phase and allocates $b$ units of resource at every subsequent step, as illustrated in Fig. \ref{fig_two_phase_allocation_with_box_constraints}.

\begin{figure}[htbp]
    \centering
    \includegraphics[width=0.6\textwidth]{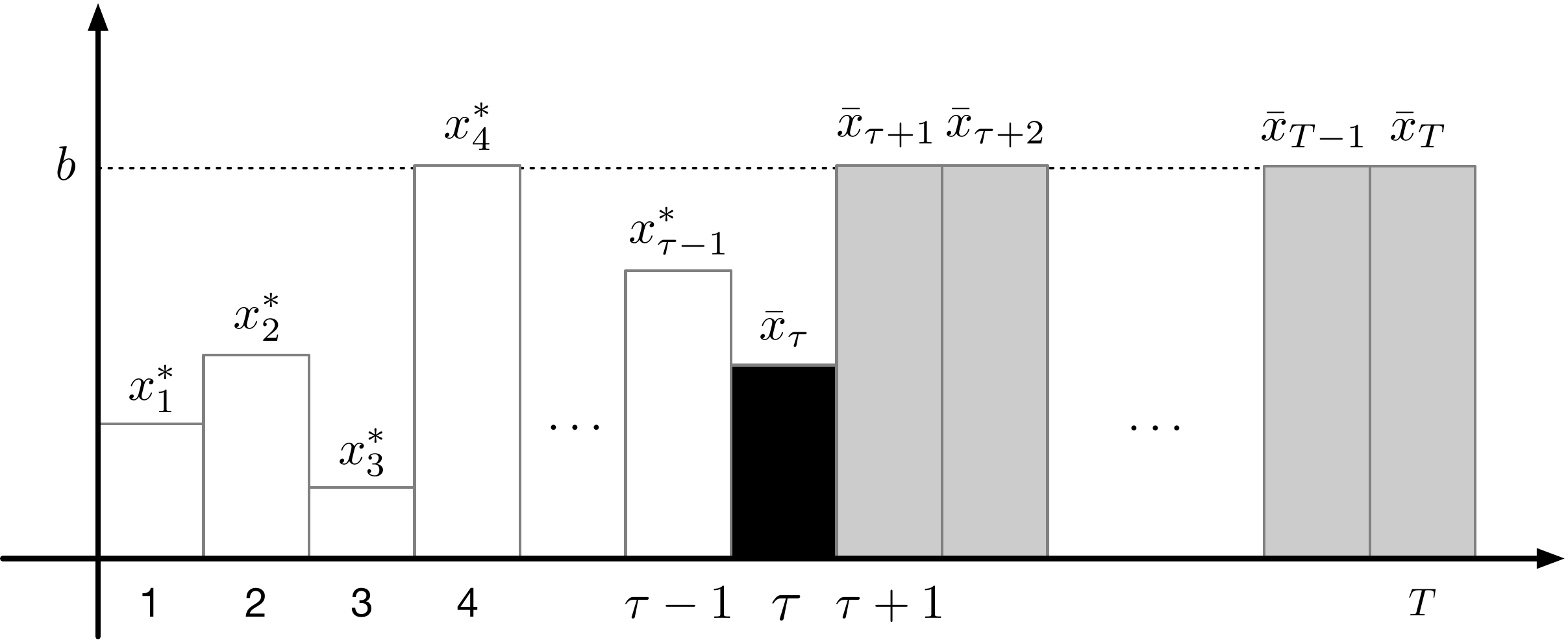}
    \caption{Illustration of the two-phase trading with box constraints. Step $\tau$ (black) separates the entire trading process into the {proactive phase} (white) and {forced phase} (gray).}
    \label{fig_two_phase_allocation_with_box_constraints}
    \vspace{-0.2cm}
\end{figure}

The value of $ b $, relative to the total resource $ k $ and the time horizon 
$ T $, determines the earliest possible point at which the algorithm's laxity will become insufficient, forcing it into the forced phase. This relationship allows us to define the \textit{earliest-possible switching step} $\tau_{\min}$, where the algorithm enters the forced phase as early as possible, depending on the system's laxity. The following lemma provides such a characterization of the lower bound on the switching step.

\begin{lemma}[\textbf{Earliest-Possible Switching Step} $ \tau_{\min} $]\label{lemma_earliest_tau}
For any $ \alpha \in [1, \theta] $, the switching step $ \tau(\alpha; \sigma) \in  \{\tau_{\min},\cdots, T\} $ holds for all  $ \sigma $, where  $ \tau_{\min} $ denotes the earliest-possible switching step, defined as
\begin{align}\label{eq:tau_min}
    \tau_{\min}=\begin{cases}
        1 & b\in (0, k/T ],\\
        T - \lceil \frac{k}{b} \rceil + 1 & b\in (k/T, k ),  \\
        T & b\in [k, +\infty ). 
    \end{cases}
\end{align}
\end{lemma}
\begin{proof}
The proof of the above two lemmas is trivial and thus is omitted for brevity. Note that for the trivial box constraint with $ b \in (0, k/T) $, the entire trading is forced from the very beginning, and thus $ \tau_{\min} = 1 $. Conversely, when $ b \in [k, +\infty) $, there is no forced phase, and the switching step occurs at $ T $.
\end{proof}

\subsection{Step 2: Reduction from Multi-Period to Single-Period Force Trading}

Recall that our goal is to find a design of $ \alpha \geq 1 $ such that the competitive ratio of Algorithm \ref{alg_RBP}, namely, $\CR(\alpha)$, is as close to 1 as possible. A direct search for the optimal $\alpha$ is not feasible since numerically evaluating $\CR(\alpha)$ will require an exhaustive computation of the switching $ \tau $ for all possible instances. However, by leveraging the structure of the worst-case instances for \OCK with box constraints, we can significantly simplify the analysis.

\begin{proposition}[\textbf{Structure of Worst-Case Instances}]\label{prop:CR_worst_p_min}
Assume that $\alpha_*$ minimizes $\CR(\alpha)$, and let the corresponding worst-case instance be denoted by $\sigma_*$. If a forced trading phase exists for the given $\sigma_*$, the corresponding switching step is defined as $ \tau_* $. In this case, the worst-case instance $\sigma_*$ takes the following form:
\begin{align*}
\sigma_* = (p_1, \ldots, p_{\tau_* - 1}, p_{\min}, \ldots, p_{\min}),
\end{align*}
where $p_t = p_{\min}$ holds for all $t > \tau_*$. 
\end{proposition}

The proof of Proposition \ref{prop:CR_worst_p_min} is given in Appendix \ref{proof_p_min}. Proposition \ref{prop:CR_worst_p_min} implies that the worst case occurs when forced trading starts at step $\tau_*$, and all remaining resource is forced to be allocated at the minimum price. Based on this result, for any case where forced trading begins at $\tau$, it suffices to focus on price sequences of types similar to $ \sigma_*$:
\begin{align}\label{omega_tau}
\sigma^{(\tau)} := (p_1, p_2, \ldots, p_{\tau-1}, p_{\min}, \ldots, p_{\min})  \in\Omega^{(\tau)},
\end{align}
where $ \Omega^{(\tau)} $ denotes the set of all input instances in which all prices are $ p_{\min} $ after step $ \tau $. Similarly, let $ \sigma_r^{(\tau)}$ be defined as:
\begin{align*}
\sigma_r^{(\tau)} := (p_1, p_2, \ldots, p_{\tau-1}, p_{\min}) \in\Omega_r^{(\tau)},
\end{align*}
where the subscript `r' is a shorthand for "reduced." 

\begin{proposition}[\textbf{Reduction to Single-Period Force Trading}]\label{prop:sigma_tau}
For any given $ b $, $ k $, and known $ T \geq 1 $, consider the price sequence $\sigma = (p_1, p_2, \ldots, p_T)$, where forced trading begins at step $\tau \in \{1, \cdots, T\}$. The optimal offline solution $ \OPT(\sigma^{(\tau)}) $, obtained by solving the original LP in Eq. \eqref{eq:OC_max} over the full horizon $T$, is equivalent to solving the following reduced LP over $ \sigma_r^{(\tau)} $, with a shortened horizon $\tau$:
\begin{subequations}\label{eq_reduce}
    \begin{align}
    \OPT(\sigma_r^{(\tau)}) & = 
    \max_{x_t \in \mathbb{R}_+} \quad \sum_{t=1}^{\tau} p_t x_t, \\
    s.t. \quad & \sum_{t=1}^{\tau} x_t \leq k; \quad x_t \leq b_t^{(\tau)}, \quad \forall t \in [\tau],
    \end{align}
\end{subequations}
where $ b_t^{(\tau)} $ is defined as follows:
\begin{align*}
b_t^{(\tau)} = 
\begin{cases} 
    b & \text{for } t \in [1, \tau-1], \\
    k & \text{for } t = \tau.
\end{cases}
\end{align*}
Namely, $ \OPT(\sigma^{(\tau)}) = \OPT(\sigma_r^{(\tau)}) $ holds for all $ \sigma^{(\tau)} \in \Omega^{(\tau)}$ and $ \sigma_r^{(\tau)} \in \Omega_r^{(\tau)}$.  
\end{proposition}

The proof of Proposition \ref{prop:sigma_tau} is provided in Appendix \ref{sec:proof_of_proposition_sigma_tau}. The high level idea for the proof is as follow: the reduction in Proposition \ref{prop:sigma_tau} enables us to focus on a simplified \OC problem, where, in the worst case, all remaining resource is allocated at $p_{\min}$ in the final step, and the box constraint is trivial (i.e., $b_{\tau}^{(\tau)} \geq k$).

\subsection{Overview of the Remaining Proof}
The remaining proof follows three key steps: 
\begin{itemize} 
    \item First, based on Proposition \ref{prop:CR_worst_p_min} and Proposition \ref{prop:sigma_tau}, we show that for any price sequence $\sigma^{(\tau)} \in \Omega^{(\tau)}$, the optimal solution $\OPT(\sigma^{(\tau)})$ can be obtained by solving the reduced LP over $  \sigma_r^{(\tau)} $, with a shortened horizon $\tau$. This reduction allows us to apply the online primal-dual analysis based on the reduced LP in Eq. \eqref{eq_reduce}, providing the upper bound below: 
    \begin{align*} 
    \frac{\OPT(\sigma^{(\tau)})}{\PRM_{\boldsymbol{\phi}}(\sigma^{(\tau)}|\alpha)}  \le  \alpha_{\tau}, \forall \sigma^{\tau} \in \Omega^{(\tau)}, \tau \in [\tau_{\min}, \tau_T],
    \end{align*} 
    where $ \alpha_{\tau} $ is defined as the root to the following equation:
    \begin{align}\label{eq:alpha_tau}
        \alpha = \tau\left[ 1 - \left( \frac{\alpha - 1}{\theta - 1} \right)^{1/\tau}\right].
    \end{align}
    This indicates that our algorithm $ \PRM_{\boldsymbol{\phi}} $ is $ \alpha_{\tau} $ competitive for solving the reduced LP. A more formal statement of the above results is given in Lemma \ref{lemma_ratio_tau} in the appendix.

    \item Second, for any switching step $\tau \in [\tau_{\min}, T]$, we prove in Lemma \ref{lemma_cr_tau} in the appendix that the following holds:
    \begin{align*} 
    \CR(\alpha) = \max_{\sigma \in \Omega} \frac{\OPT(\sigma)}{\PRM_{\boldsymbol{\phi}}(\sigma|\alpha)} \le \alpha, \text{ } \forall \alpha \in [\alpha_{\tau}, \alpha_{\tau+1}),
    \end{align*}
    indicating that our algorithm $ \PRM_{\boldsymbol{\phi}} $ is also $ \alpha $-competitive for solving the original LP.

    \item Finally, since $\alpha_{\tau}$ increases with $\tau$ (by Eq. \eqref{eq:alpha_tau}), we have
    \begin{align*} 
    \CR^*_{\textsf{known}} = \min_{\alpha \in [\alpha_{\tau_{\min}}, \alpha_T]} \max_{\sigma \in \Omega} \frac{\OPT(\sigma)}{\PRM_{\boldsymbol{\phi}}(\sigma|\alpha)} \le \alpha_{\tau_{\min}}. 
    \end{align*}
\end{itemize}
Eq. \eqref{eq:CR_known} follows by substituting $ \tau_{\min} = T - \lceil \frac{k}{b} \rceil + 1 $ into Eq. \eqref{eq:alpha_tau}. We thus complete the proof of Theorem \ref{thm:oc_known_notice} for \OCK with non-trivial box constraints. 

\section{\OCP: Learning-Augmented Algorithm with Horizon Prediction}

Building on our unified algorithm, we extend to \OCP by incorporating horizon predictions. The goal is to develop an algorithm that adapts between \OCK and \OCU: When predictions are accurate, the algorithm's performance should approach that of \OCK. However, if the predictions are inaccurate, the algorithm should revert to the robustness of \OCU.

\subsection{Definitions and Core Ideas}
To formalize this adaptive approach, we introduce the concepts of robustness and consistency within the \OC framework. Let $ \ALG(\sigma) $ be the objective value produced by an online algorithm $ \ALG $ on an instance $ \sigma $. Similarly, let $ \OPT(\sigma) $ denote the objective value of the optimal offline solution.

\textbf{Definition of robustness and consistency.}  
Given a confidence parameter $ \lambda \in [0, 1] $, an online algorithm $ \ALG $ is $ \eta(\lambda) $-consistent and $ \gamma(\lambda) $-robust if, with a given horizon prediction $ T_{\textsf{pred}} $ and real horizon $ T $, for the instance $ \sigma = \sigma_{\textsf{pred}} = (p_1, \ldots, p_{T_{\textsf{pred}}}) $, $\frac{{\OPT}(\sigma_{\textsf{pred}})}{{\ALG}(\sigma_{\textsf{pred}})} \leq \eta(\lambda)$ and for all instances $ \sigma \in \Omega$, $\frac{\OPT(\sigma)}{{\ALG}(\sigma)} \leq \gamma(\lambda)$. 

We aim to design algorithms such that $ \eta(\lambda) $ approaches 1 as $ \lambda $ approaches 0, meaning the algorithm performs nearly as well as the offline optimal solution when predictions are accurate. Additionally, $ \gamma(\lambda) $ should approach the best possible guarantee in \OCU as $ \lambda $ approaches 1.

\textbf{Core idea of Algorithm \ref{alg_RBP_combine}.}  
The core idea of the learning-augmented algorithm is to combine the pseudo-cost functions of \OCK and \OCU to balance robustness and consistency. We split the resource $ k $ into two parts: $ k_0^{(1)} = (1-\lambda)k $ for \OCK and $ k_0^{(2)} = \lambda k $ for \OCU, where $ \lambda \in [0, 1] $.  In Algorithm \ref{alg_RBP_combine}, the pseudo-cost function $ \phi_1 $ in Eq. \eqref{eq:argmax_k1}  follows Definition \ref{def_price} with $ k=k_0^{(1)}$, and $ \phi_2 $ in Eq. \eqref{eq:argmax_k2} is defined as follows:
\begin{align*}
\phi_2(\beta| {F}_t^{(2)},\alpha_2)=
\begin{cases}
p_{\min} &  c_{t-1} \in \left[0, \frac{k_0^{(2)}}{\alpha_2}\right), \\
p_{\min}e^{\frac{\alpha_2}{k_0^{(2)}}(\beta+c_{t-1})-1} &  
 c_{t-1} \in [\frac{k_0^{(2)}}{\alpha_2},k_0^{(2)}],
\end{cases}
\end{align*}
where $ c_{t-1} = \sum_{i=1}^{t-1}x_i^{(2)} $,  which follows the same design for handling \OCU (See Appendix~\ref{appendix_OC_unknown} for more detail). Algorithm~\ref{alg_RBP_combine} runs Algorithm~\ref{alg_RBP} for \OCK with $k_0^{(1)}$ and for \OCU with $k_0^{(2)}$ in parallel, combining trading decisions from both algorithms.

\subsection{Robustness-Consistency Analysis}
Algorithm \ref{alg_RBP_combine} optimizes the division between \OCK and \OCU based on $ \lambda $. Theoretical guarantees are provided in Theorem \ref{thm:combine_cr}, showing robustness $\gamma(\lambda)$ and consistency $\eta(\lambda)$ bounds.

\begin{theorem}\label{thm:combine_cr}
For a given prediction $T_{\textsf{pre}}$ with predefined $\lambda$, running Algorithm \ref{alg_RBP_combine} with $\alpha_1 = \CR_{\textsf{known}}^*$ and $\alpha_2 = \CR_{\textsf{unknown}}^* $ provides the following upper bounds on robustness $\gamma(\lambda)$ and consistency $\eta(\lambda)$:
\begin{itemize}
    \item For robustness: $ \gamma(\lambda) \le \frac{\CR_{\textsf{unknown}}^*}{\lambda}$. 
    \item For consistency:
    \begin{align*}
    \eta(\lambda) \le \frac{\CR_{\textsf{known}}^*\cdot \CR_{\textsf{unknown}}^*}{\CR_{\textsf{unknown}}^* + \lambda(\CR_{\textsf{known}}^* - \CR_{\textsf{unknown}}^*)},
    \end{align*}
\end{itemize}
where $ \CR_{\textsf{known}}^* $ is given by Eq. \eqref{eq:CR_known} and $ \CR_{\textsf{unknown}}^* = 1 + \ln\theta$.
\end{theorem}

The proof of Theorem \ref{thm:combine_cr} is given in Appendix \ref{appendix_la}. Theorem \ref{thm:combine_cr} establishes a clear relationship between the consistency parameter $\eta(\lambda)$ and the robustness parameter $\gamma(\lambda)$ in relation to the confidence level $\lambda$. When $\lambda$ is small, indicating a high degree of trust in the horizon prediction, the algorithm's consistency is maximized, with $\eta(\lambda)$ approaching $\CR_{\textsf{known}}^*$, signifying near-optimal performance, as if the real horizon were known in advance. This trade-off allows the algorithm to dynamically balance between optimizing for consistency when the prediction is trusted and safeguarding robustness when the prediction is unreliable, achieving strong performance in both cases.

To validate our theoretical findings, we conduct a case study in Appendix \ref{appendix_case_study}, focusing on an energy trading problem. 

\begin{algorithm}[t]
\caption{Algorithm for \OCP}
\label{alg_RBP_combine}
\begin{algorithmic}[1]
    \STATE \textbf{Inputs:} $T_{\textsf{pred}}$ (prediction of the horizon $T$), $\lambda \in [0 , 1]$ (the confidence level), $b$ (box-constraint), $\alpha_1$, and $\alpha_2$.
    \STATE Define $k_0^{(1)} = (1 - \lambda)k$ and $k_0^{(2)} = \lambda k$. 
    \STATE Define $b^{(1)} = (1 - \lambda)b$ and $b^{(2)} = \lambda b$.
    
    \WHILE{price $p_t$ is revealed}
        \STATE Calculate $x_t^{(1)}$:
        \begin{align}\label{eq:argmax_k1}
        x_t^{(1)} = \argmax_{x_t \in [0, b^{(1)}]} \left( p_t x_t - \int_{0}^{x_t} \phi_1(\beta|  {F}_t^{(1)},\alpha_1) d\beta \right).
        \end{align}
        \IF{$t > T_{\textsf{pred}}$}
            \STATE $\bar{x}_t^{(1)} = 0$ 
        \ELSIF{$k_{t-1}^{(1)} - x_t^{(1)} \leq b^{(1)}(T_{\textsf{pred}} - t)$} 
            \STATE $\bar{x}_t^{(1)} = \min\{x_t^{(1)}, k_{t-1}^{(1)}\}$ 
        \ELSE
            \STATE $\bar{x}_t^{(1)} = \min\{b^{(1)}, k_{t-1}^{(1)} - b^{(1)}(T_{\textsf{pred}} - t)\}$ 
        \ENDIF
        
        \STATE Calculate $x_t^{(2)}$:
        \begin{align}\label{eq:argmax_k2}
        x_t^{(2)} = \argmax_{x_t \in [0, b^{(2)}]} \left( p_t x_t - \int_{0}^{x_t} \phi_2(\beta|  {F}_t^{(2)},\alpha_2) d\beta \right).
        \end{align}
        \STATE Trade $\bar{x}_t$ units at price $p_t$ according to:
        \begin{align*}
            \bar{x}_t = \bar{x}_t^{(1)} + x_t^{(2)}.
        \end{align*}
        
        \STATE Update the resource for $k^{(1)}$: $k_t^{(1)} = k_{t-1}^{(1)} - \bar{x}_t^{(1)}$.
    \ENDWHILE
\end{algorithmic}
\end{algorithm}

\section{Conclusion and Future Work}
In this paper, we formulated the online conversion (\OC) problem and presented a unified algorithm to address it under three horizon uncertainty models: known, notified, and unknown horizons. The algorithm handles scenarios with and without rate limits, achieving tight robustness guarantees with optimal competitive ratios. We also proposed a learning-augmented algorithm that leverages horizon predictions to improve performance while maintaining robustness in uncertain conditions. 

Future work includes extending \OC to bi-directional trading, addressing the two-way online conversion problem, which involves optimizing cost minimization and revenue generation simultaneously in dynamic markets. This remains an open and promising challenge with significant theoretical and practical implications. In addition, we believe the idea of designing an appropriate pseudo-cost to guide the online decision-making process will find other applications in related problems. 

\printbibliography{}

\newpage
\appendix
{\LARGE{\textbf{Appendix}}}

\section{Case Study: Energy Trading}\label{appendix_case_study}
In this section, we empirically validate our theoretical results within the context of energy trading and demonstrate the performance of Algorithms \ref{alg_RBP} and \ref{alg_RBP_combine} under real market conditions across the four \OC scenarios: \OCK, \OCN, \OCU, and \OCP.

\subsection{Experimental Setup}
This section outlines the experimental setup used to evaluate the performance of our algorithms in the vehicle-to-grid (V2G) energy trading context.  We focus on the Tesla Model 3 for its well-documented specifications. The battery capacity ($k$) is set at 68 kWh in our simulations, reflecting the upper range of usable capacity (45–68 kWh, accounting for daily driving reserves) \cite{evdatabase_tesla_model_3}. The discharging rate is set to 72 kWh/h (or 6 kWh per 5-minute slot), aligned with the technical capabilities of the Tesla Model 3, which, although not officially equipped for bidirectional charging, is suitable for V2G applications. Market prices are derived from PJM frequency regulation data \cite{PJMRegulationPrices}, recorded at 5-minute intervals from August 11 to August 30, 2024 (i.e., 20 days in total). Each day is divided into 288 slots (12 per hour). Price bounds are set at $[p_{\min}, p_{\max}]$, with $p_{\min} = 5$ and $p_{\max} = 1000$, reflecting realistic V2G price fluctuations.

\subsection{Experimental Results}
The key performance metric we use is the \textit{Empirical Competitive Ratio} (\ECR), which measures the performance ratio of our algorithm over the offline optimal counterpart for a price sequence $ \sigma $. We then calculate the average \ECR over 20 days to identify trends in algorithm performance across the three \OC settings (\OCK, \OCN, \OCU) and to assess the effectiveness of the learning-augmented algorithm under \OCP with varying predictions of  $T$.

\paragraph{\bfseries Evaluation of the Unified Algorithm}
Figure \ref{fig:oc_three_settings} shows a comparison of the unified algorithm across these settings with $ b $, using $\theta = 200$ from the dataset. In Figure \ref{fig:sub_three_settings_1}, average \ECR is plotted against varying $T$ with fixed $b$, while Figure \ref{fig:sub_three_settings_2} shows \ECR as $b$ changes within $ (0, 60]$ (kWh/slot), and Figure \ref{fig:sub_three_settings_3} presents the cumulative distribution of \ECR. 

Based on Figure \ref{fig:oc_three_settings}, several key observations emerge. As $T$ increases, \OCU performance improves while \OCK and \OCN decline, narrowing the gap between them. This is because, with larger $T$, the forced trading phase in \OCK and \OCN becomes less relevant, leading to similar behaviors in all settings. Additionally, the algorithm consistently performs better as $b$ increases, suggesting that maximizing $b$ for energy trading is beneficial in practice. \OCK and \OCN show similar performance across all subfigures, indicating that receiving a timely notification can be as effective as knowing the total horizon $T$, while \OCU always performs the worst among the three settings due to a lack of horizon information.

\begin{figure*}
    \centering
    \begin{subfigure}[b]{0.32\textwidth}
        \centering
        \includegraphics[width=\textwidth]{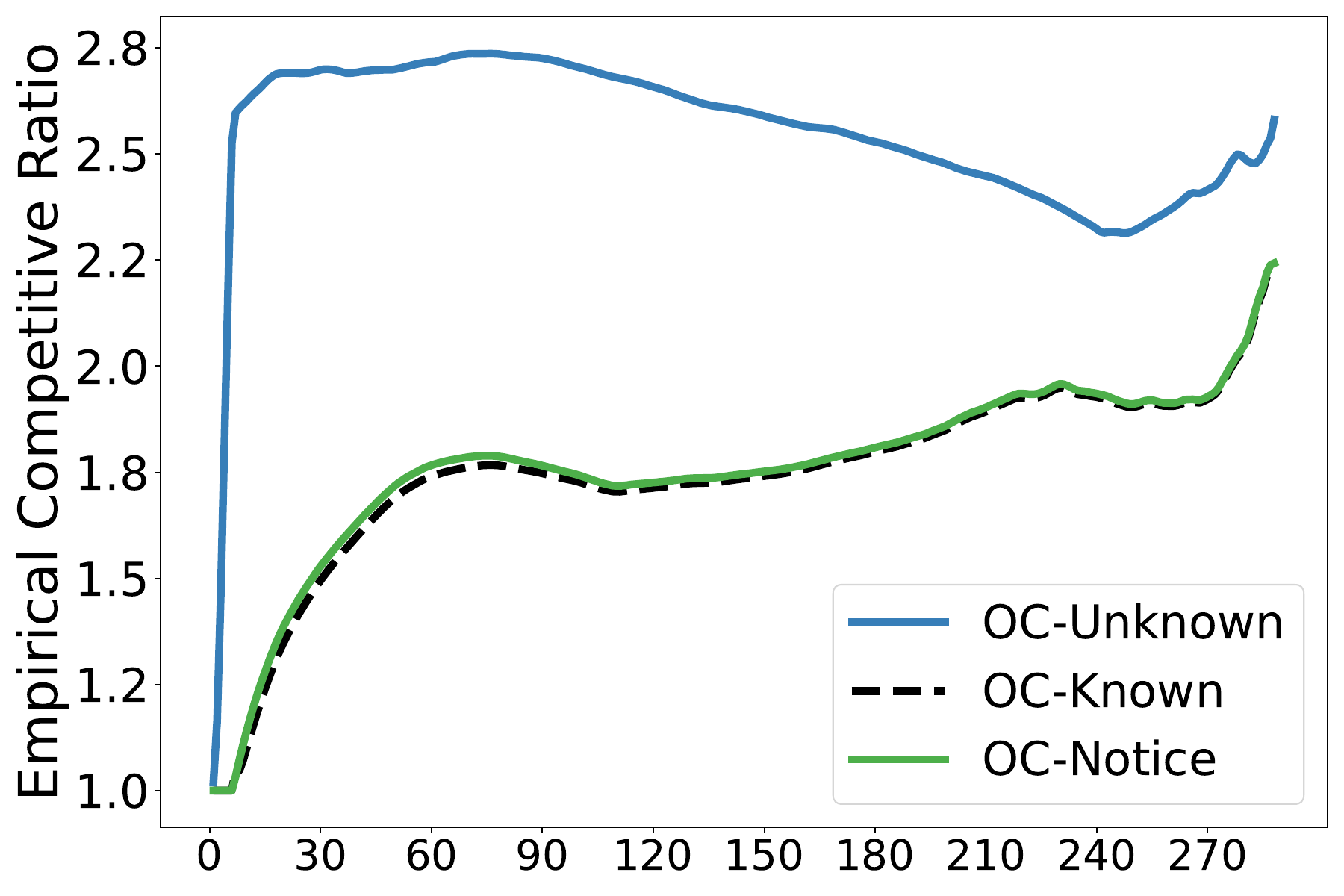}
        \caption{Averaged \ECR over time horizon $T$ (time slots)}
        \label{fig:sub_three_settings_1}
    \end{subfigure}
    \hfill
    \begin{subfigure}[b]{0.32\textwidth}
        \centering
        \includegraphics[width=\textwidth]{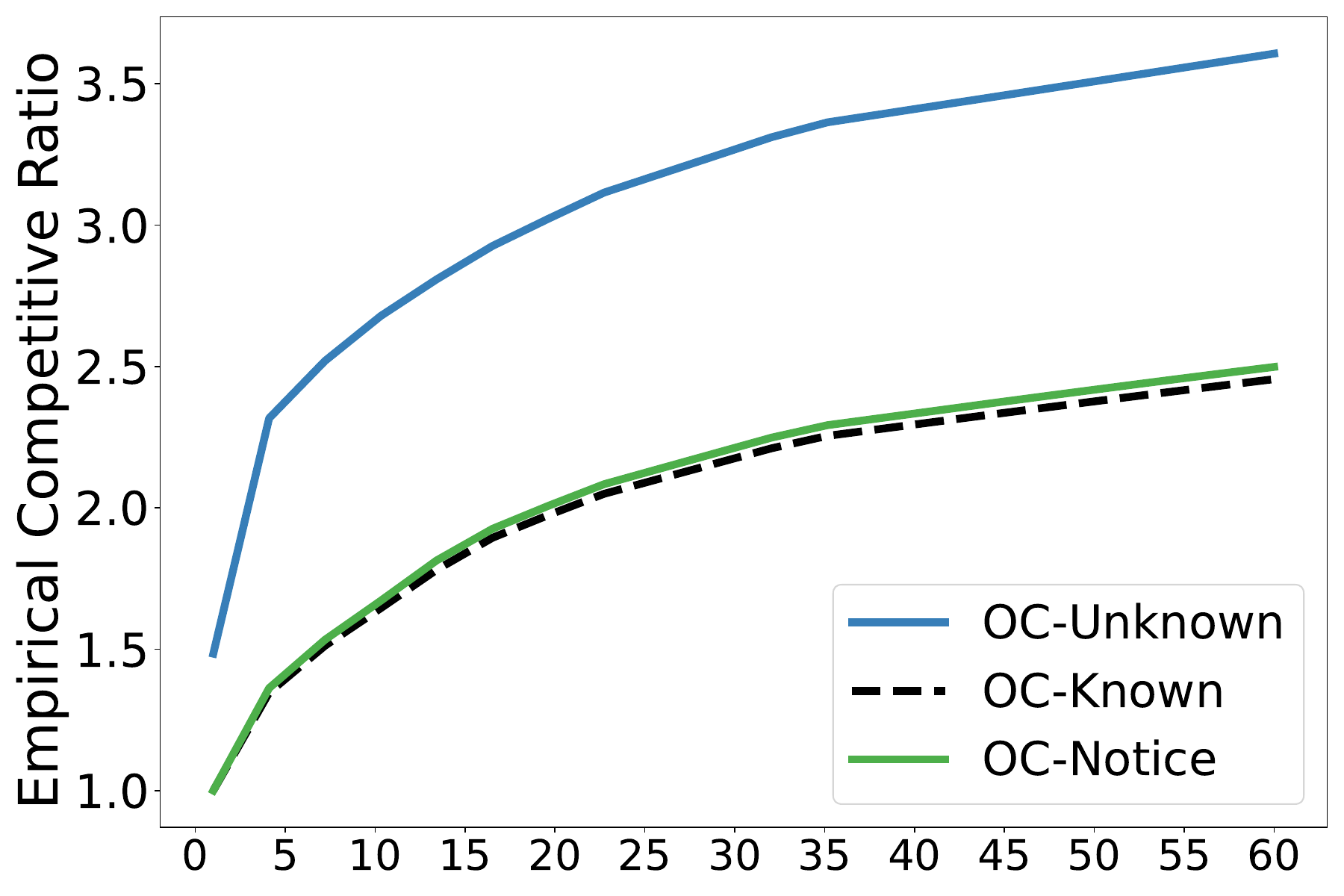}
        \caption{Averaged \ECR vs rate limit $b$ (kWh/slot)}
        \label{fig:sub_three_settings_2}
    \end{subfigure}
    \hfill
    \begin{subfigure}[b]{0.32\textwidth}
        \centering
        \includegraphics[width=\textwidth]{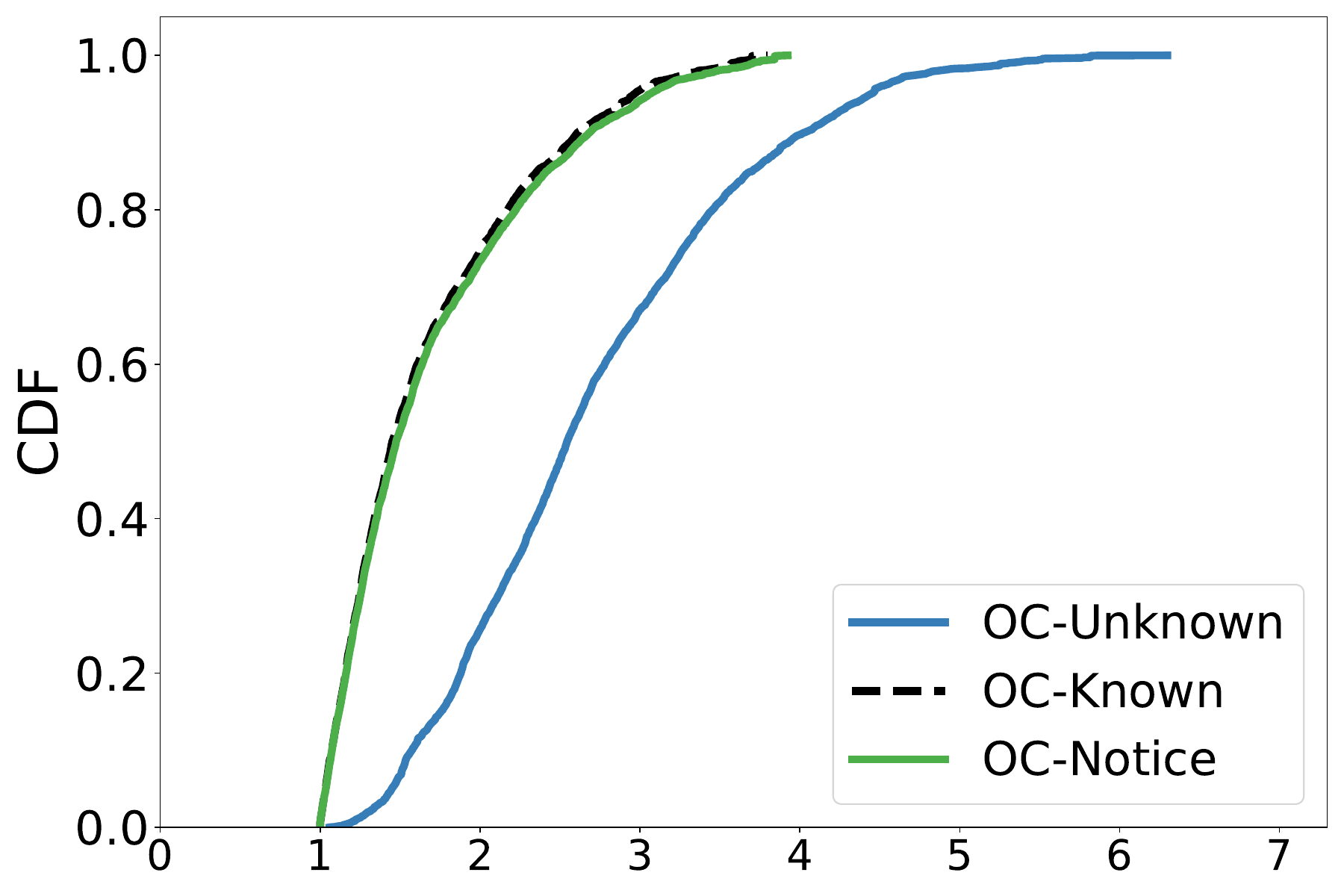}
        \caption{CDF plot vs \ECR  to check the distribution over horizon settings}
        \label{fig:sub_three_settings_3}
    \end{subfigure}
    \caption{Comparison of the unified algorithm across three settings with rate limit.}
    \label{fig:oc_three_settings}
\end{figure*}

\begin{figure*}
    \centering
    \begin{subfigure}[b]{0.32\textwidth}
        \centering
        \includegraphics[width=\textwidth]{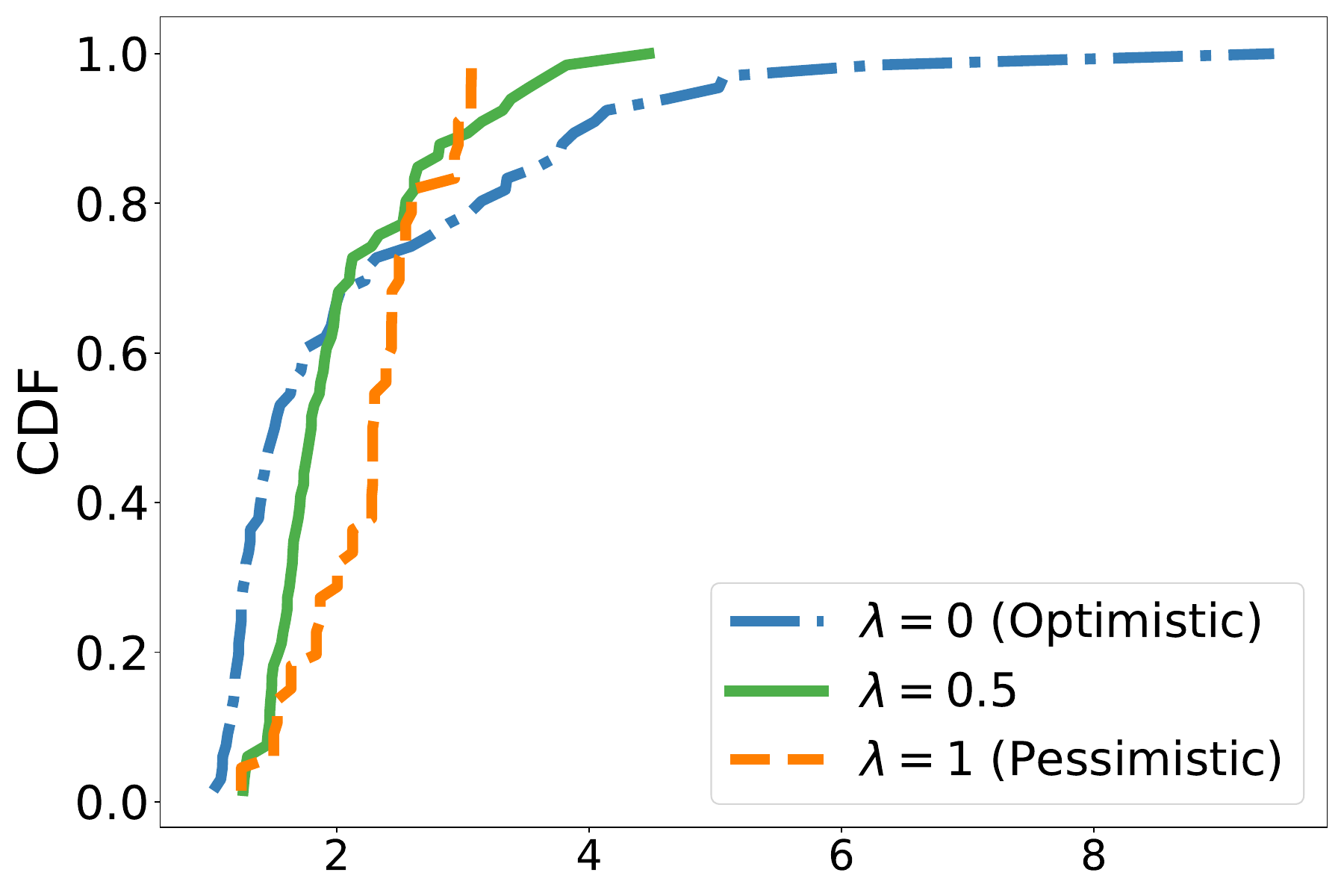}
        \caption{CDF plot vs \ECR ($T_{\textsf{pred}}=0.5T$)}
        \label{fig:cdf_ECR_lambda_predictions}
    \end{subfigure}
    \hfill
    \begin{subfigure}[b]{0.32\textwidth}
        \centering
        \includegraphics[width=\textwidth]{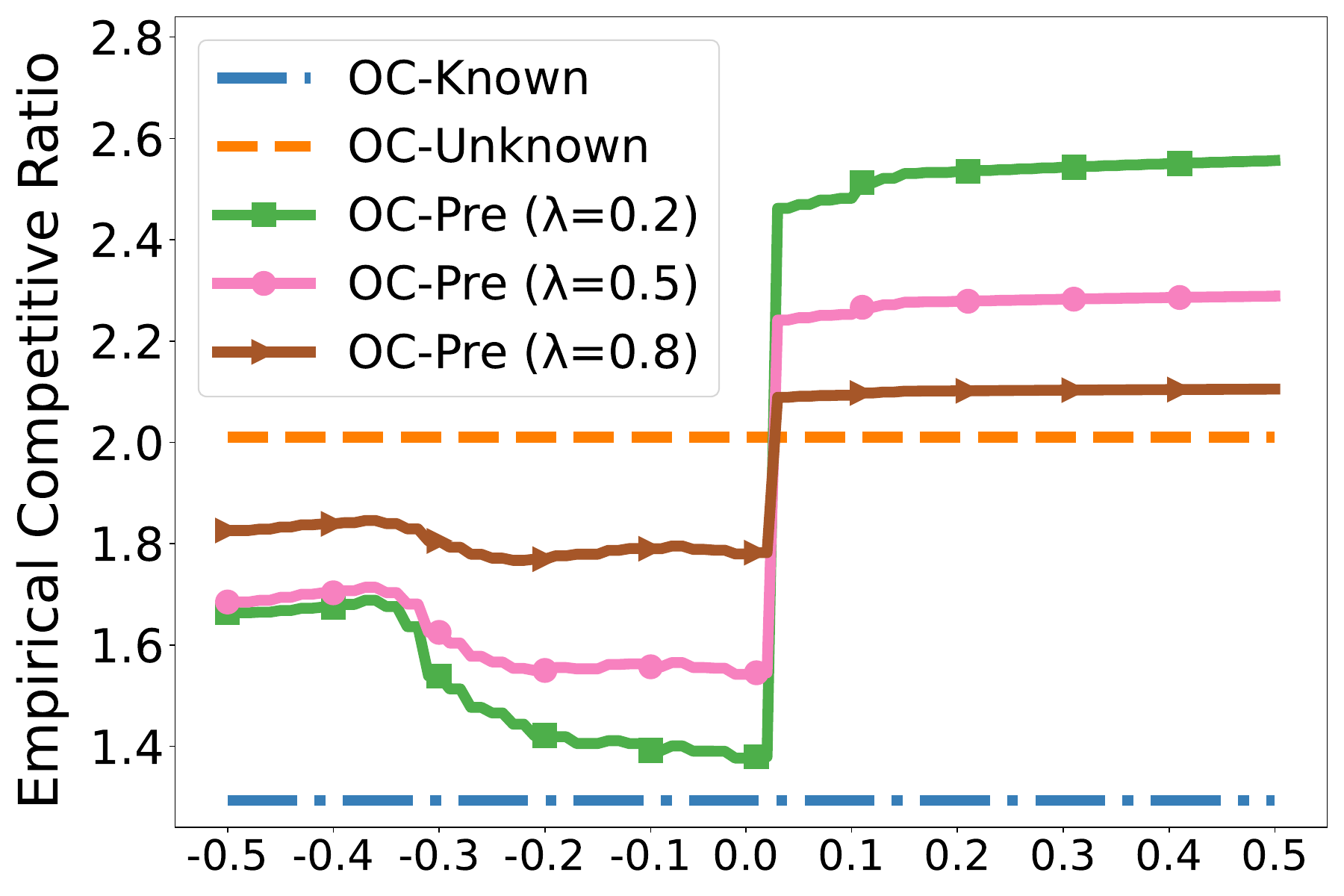}
        \caption{\ECR vs predictions error}
        \label{fig:ECR_predictions}
    \end{subfigure}
    \hfill
    \begin{subfigure}[b]{0.32\textwidth}
        \centering
        \includegraphics[width=\textwidth]{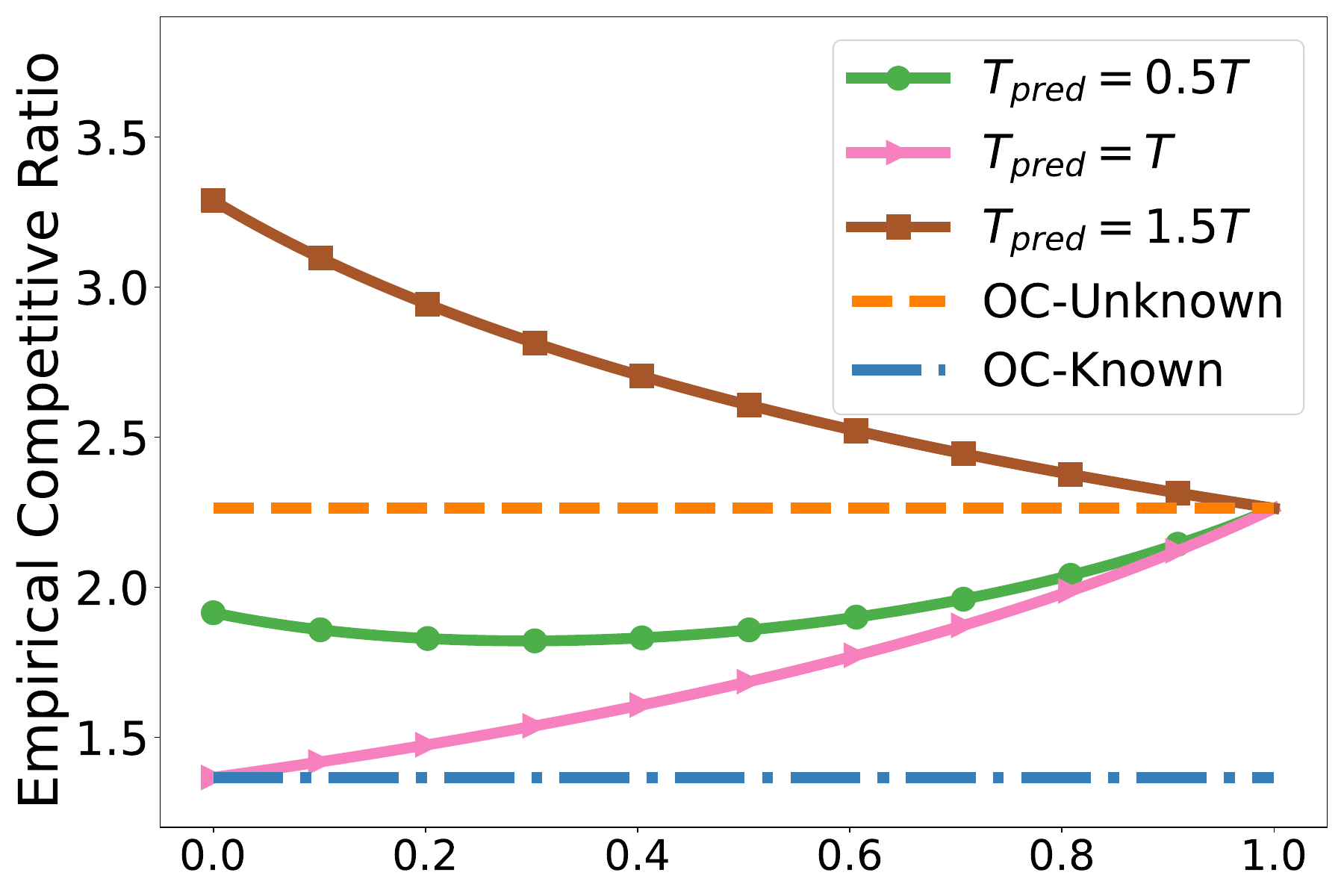}
        \caption{\ECR vs $\lambda$}
        \label{fig:ECR_lambda_predictions}
    \end{subfigure}
    \caption{Comparison of the learning-augmented algorithm across different settings of $T_{\textsf{pred}}$ and $\lambda$.}
    \label{fig:plots_summary_midnight}
\end{figure*}

\begin{figure}
    \centering
    \vspace{-0.2cm} 
    \begin{subfigure}[b]{0.9\textwidth}
        \centering
        \includegraphics[width=\textwidth]{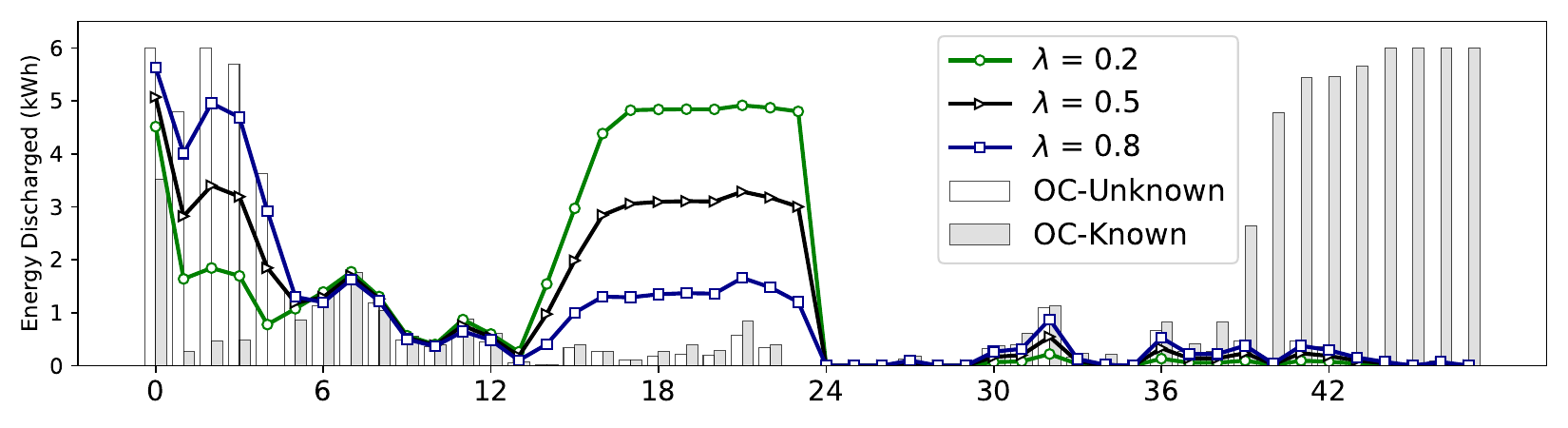}
        \caption{Averaged energy traded vs time slot $t$ (under prediction: $T_{\textsf{pred}} = 0.5T$)}
        \label{fig:energy_discharge_key1}
    \end{subfigure}
    \vspace{0.2cm} 
    \begin{subfigure}[b]{0.9\textwidth}
        \centering
        \includegraphics[width=\textwidth]{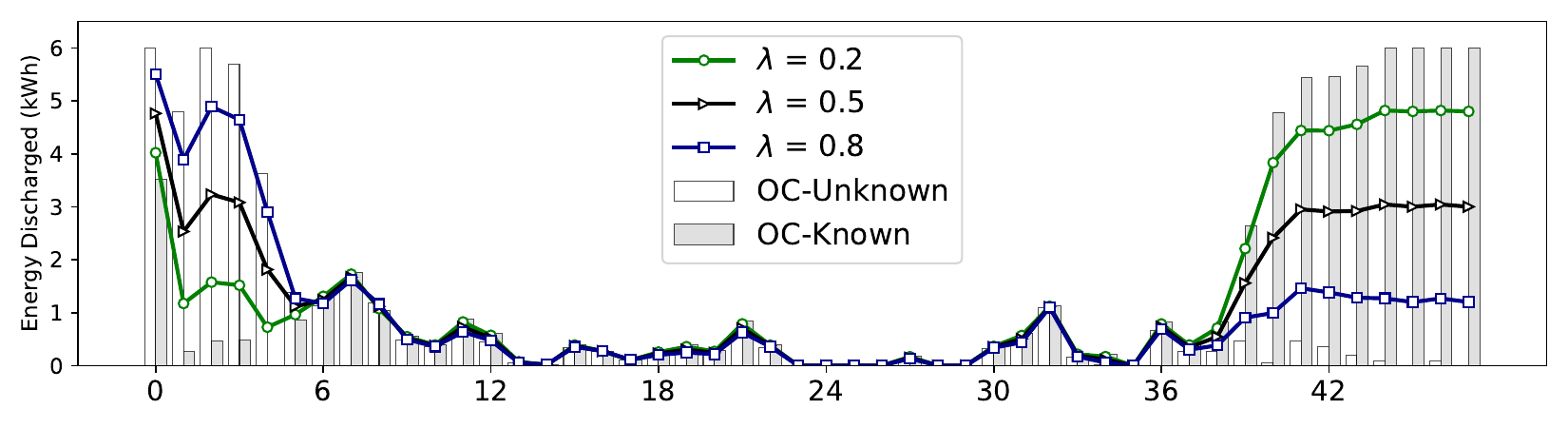}
        \caption{Averaged energy traded vs time slot $t$ (accurate prediction: $T_{\textsf{pred}} = T$)}
        \label{fig:energy_discharge_key2}
    \end{subfigure}
    \vspace{0.2cm} 
    \begin{subfigure}[b]{0.9\textwidth}
        \centering
        \includegraphics[width=\textwidth]{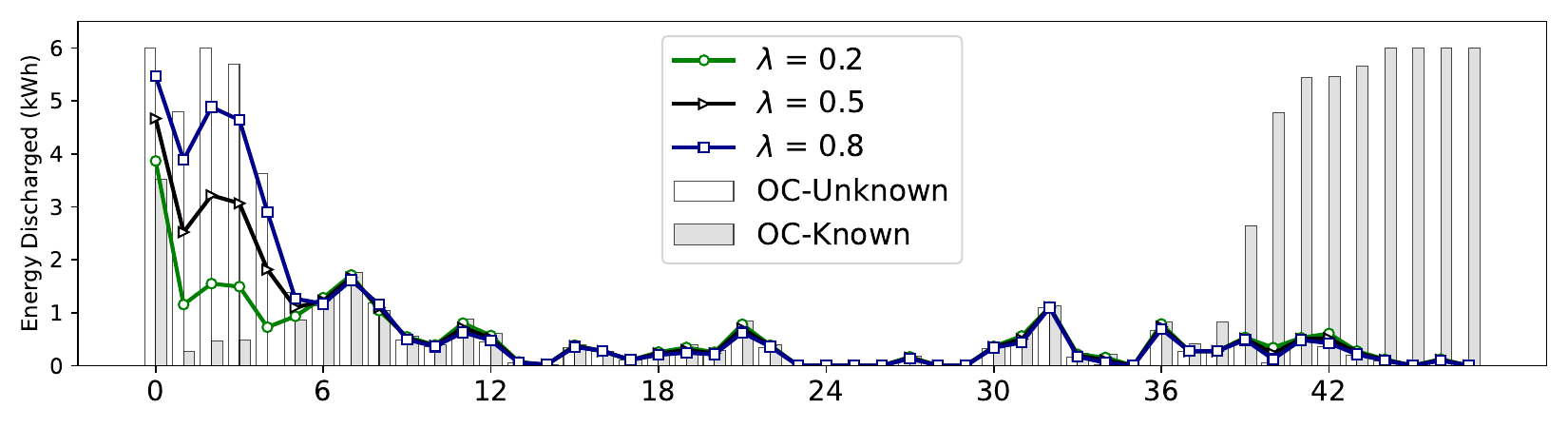}
        \caption{Averaged energy traded vs time slot $t$ (over prediction: $T_{\textsf{pred}} = 1.5T$)}
        \label{fig:energy_discharge_key3}
    \end{subfigure}
    \caption{Performance of Algorithm \ref{alg_RBP} under various $\lambda$ and $T_{\textsf{pred}}$.}
    \label{fig:plots_summary_midnight_1}
\end{figure}

\paragraph{\bfseries Evaluation of the Learning-Augmented Algorithm} Figure \ref{fig:plots_summary_midnight} summarizes the \ECR results for different horizon predictions and values of $\lambda$.

In Figure \ref{fig:cdf_ECR_lambda_predictions}, we compare \OCP (with $T_{\textsf{pred}}=0.5T$) across three settings of $\lambda$ ($[0, 0.5, 1]$) using a CDF plot. The learning-augmented algorithm demonstrates solid average performance across all settings. While the optimistic case ($\lambda=0$) shows a long tail, indicating some risk, it still provides a reasonable worst-case guarantee, balancing both robustness and consistency effectively.

Figure \ref{fig:ECR_predictions} presents the \ECR trends for \OCK, \OCU, and \OCP (with varying $\lambda$) against prediction error, defined as $ \frac{T_{\textsf{pred}} - T}{T}$. A negative error means the horizon is underestimated, while a positive error indicates overestimation. At zero error ($T_{\text{pred}} = T$), a dramatic shift is observed. On the left side (underestimation), the performance of \OCP with increasing $\lambda$ improves, converging toward \OCK. Conversely, on the right side (overestimation), performance worsens as $\lambda$ decreases, converging toward \OCU. Interestingly, on the right side, increasing prediction error does not seem to affect performance significantly, which implies the robustness of the algorithm. On the left side, larger prediction errors degrade \OCP's performance, moving it closer to \OCU. Notably, for large negative errors, \OCP with smaller $\lambda$ performs similarly to $\lambda = 0.5$, as the algorithm struggles to improve with reduced horizon predictions. However, the gap between small and larger $\lambda$ (e.g., $\lambda = 0.8$) persists, suggesting that placing too little trust in predictions makes the algorithm less sensitive to smaller horizons.

In Figure \ref{fig:ECR_lambda_predictions}, the average \ECR for varying $\lambda$ under different horizon predictions is shown. As $\lambda$ approaches 1, performance converges to \OCU, while for $\lambda = 0$, the algorithm behaves similarly to \OCK when predictions are accurate. In cases of horizon overestimation, performance worsens, but improves as $\lambda$ increases, converging toward \OCU. This implies that overly small $\lambda$ values should be avoided. For horizon underestimation, results form a curve, with the best performance occurring at intermediate $\lambda = 0.4$, suggesting that a balanced choice of $\lambda$ can s help enhance algorithm performance.

\paragraph{\bfseries Comparison of Energy Trading Profiles} 
Figure \ref{fig:plots_summary_midnight_1} shows the average energy traded (in kWh) at each time slot for different values of $\lambda$ and predicted horizons $T_{\textsf{pred}}$. Lines represent decisions based on Algorithm \ref{alg_RBP_combine}, while white and grey bars indicate the behavior under \OCU and \OCK settings, respectively. The subfigures compare performance across predictions of $T$. The over-predicting case (Figure \ref{fig:energy_discharge_key3}) with $ T_{\textsf{pred}} =1.5 T$ leads to poor performance, as the algorithm reserves energy for future high prices that never occur, missing better prices earlier. In contrast, under-predicting case (Fig. \ref{fig:energy_discharge_key1}) with $ T_{\textsf{pred}} =0.5 T$ results in better performance by capitalizing on good prices earlier. Within each subfigure, the impact of $\lambda$ is also evident. As $\lambda$ increases, the algorithm relies less on the predicted $T$, shifting performance from \OCK (closer to grey bars) to \OCU (closer to white bars).

\paragraph{\bfseries Comparison of Cumulated Revenue}
To further understand how the learning-augmented algorithm operates, Figure \ref{fig:oc_across_T_bar} compares the cumulated revenue across different predictions of $T$, with $ \lambda = 0.5 $. Notably, while the algorithm with underestimated $T$ generates good revenue early on, the one with accurate predictions ultimately achieves the highest revenue. In contrast, overestimation consistently results in the worst performance.

\begin{figure}[htbp]
    \centering
    \includegraphics[width=0.6\textwidth]{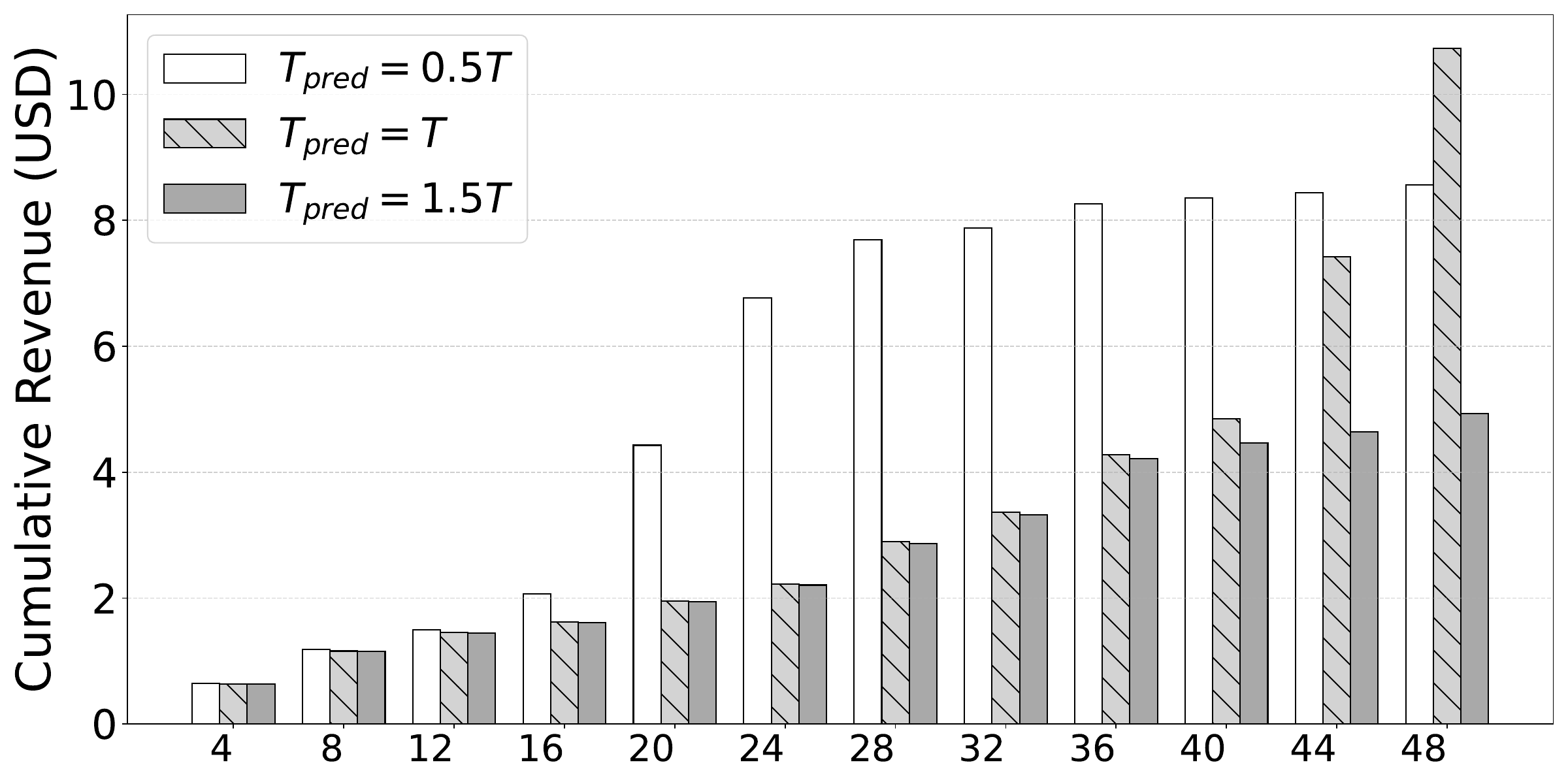}
    \caption{Cumulative revenue for the learning-augmented algorithm vs time slot $t$. This figure compares the cumulative revenue over time under three different predictions of $T$, with $\lambda = 0.5$. The plot highlights how revenue accumulation changes based on the accuracy of predicted time horizons, particularly when the predicted $T$ is larger than the actual horizon.}
    \label{fig:oc_across_T_bar}
\end{figure}

\section{Proof of Theorem \ref{thm:oc_known_notice} for \OCK and \OCN without Box Constraints}
\label{sec:proof_of_thm_known_notice}

This section presents the proofs for Theorem \ref{thm:oc_known_notice} \textit{without} considering box constraints. The remaining proofs for \OCK and \OCN \textit{with} box constraints are provided in Appendix \ref{sec:proof_of_theorem_CR_box_constraint} and Appendix \ref{sec:proof_of_theorem_CR_box_constraint_notice}, respectively. For simplicity of notation, we omit $ ( {F}_t,  \alpha) $ and directly write the pseudo-cost function $ \phi_t(x_t|  {F}_t,\alpha) $ as $ \phi_t(x_t) $.

\subsection{Some Useful Lemmas of the Pseudo-Cost Function}
\label{sec:property_of_phi}
Based on the definition of $ \bm{\phi} $ in Eq. \eqref{eq_def_phi_t}, the increment from $  \phi_{t-1}(x_{t-1}) $ to $  \phi_t(x_t) $ can be related to the difference between $  \phi_t(x_t) $ and $ p_t $. This property is formally stated  in Lemma \ref{lem:difference_phi} below.

\begin{lemma} 
\label{lem:difference_phi}
With given $ \alpha $ and $  {F}_t = ( {x}_{t-1}, \sigma_t)$, the difference between $  \phi_t(x_t) $ and  $  \phi_{t-1}(x_{t-1}) $ satisfies:
\begin{align*}
     \phi_t(x_t) - \phi_{t-1}(x_{t-1}) = \frac{\alpha}{k} x_t[ \phi_t(x_t ) - p_{\min}] .
\end{align*}
\end{lemma}
\begin{proof}
Based on Eq. \eqref{eq_def_phi_t}, we have
\begin{align*}
    & \phi_t(x_t) - \phi_{t-1}(x_{t-1}) \\
   =\ & \left[\frac{1}{1-\frac{\alpha}{k}  x_t} - 1\right] \cdot \frac{ (\alpha - 1) p_{\min} }{\prod_{i=1}^{t-1}(1 - \frac{\alpha}{k} x_i)}\\
   =\ & \frac{\alpha}{k} x_t \cdot  \frac{(\alpha - 1) p_{\min}}{\prod_{i=1}^{t}(1 - \frac{\alpha}{k}  x_i)}\\
   =\ &\frac{\alpha}{k} x_t [\phi_t(x_t) -p_{\min}].
\end{align*}
Lemma \ref{lem:difference_phi} thus follows.
\end{proof}

Another useful property is that,  when running Algorithm \ref{alg_RBP} with pseudo-cost function $ \bm{\phi} $, there is always an achievable $ x^*_t $ which satisfies $ p_t = \phi_t(x^*_t)$ at step $ t $, provided that $ p_t > \phi_{t-1}(x^*_{t-1}) $ holds.

\begin{lemma} 
\label{lemma_achievable_price}
When running Algorithm \ref{alg_RBP} with pseudo-cost function $ \bm{\phi} $ and $ \alpha>1 $, at each time $ t\in[T] $, if $ p_t > \phi_{t-1}(x^*_{t-1}) $, then there is always 
\begin{align}\label{eq_alg_x_t}
x^*_t = \frac{k}{\alpha } \left( 1-\frac{ \phi_{t-1}(x^*_{t-1})-p_{\min}}{p_t-p_{\min}} \right).
\end{align} 
\end{lemma}
\begin{proof} 
By the first-order optimality condition of Eq. \eqref{pseudo_revenue},
\begin{align*}
p_t = \ &  \phi_t(x^*_t).
\end{align*}
Thus, we have
\begin{align*}
p_t=\ & p_{\min} + \frac{\alpha p_{\min} - p_{\min}}{(1- \frac{\alpha}{k}  x^*_t) \prod_{i=1}^{t-1}(1 - \frac{\alpha}{k}  x_i^*)},
\end{align*}
which leads to
\begin{align*}
1-\frac{\alpha}{k}  x^*_t =\ & \frac{1}{p_t-p_{\min}} \cdot \frac{\alpha p_{\min} - p_{\min}}{\prod_{i=1}^{t-1}(1-\frac{\alpha}{k}  x_i^*)}\\
=\ & \frac{1}{p_t-p_{\min}} \cdot ( \phi_{t-1}(x^*_{t-1})-p_{\min}).
\end{align*}
Solving the above equation leads to
\begin{align*}
x^*_t =\ & \frac{1-( \phi_{t-1}(x^*_{t-1})-p_{\min})/(p_t-p_{\min})}{\alpha  /k} \nonumber \\
=\ & \frac{k}{\alpha }  \left( 1-\frac{ \phi_{t-1}(x^*_{t-1})-p_{\min}}{p_t-p_{\min}} \right).
\end{align*}
Note that $ p_t \ge  \phi_{t-1}(x^*_{t-1})> p_{\min} $, so $ \frac{ \phi_{t-1}(x^*_{t-1})-p_{\min}}{p_t-p_{\min}}\in (0,1] $, which means that given $ \alpha >1 $, there is always an achievable $ x_t^*\in [0, \frac{1}{\alpha }) $ that maximizes Eq. \eqref{pseudo_revenue} with $ p_t= \phi_{t}(x^*_{t}) $.  Therefore, Lemma \ref{lemma_achievable_price} follows.
\end{proof}

Combining the results of Lemma \ref{lem:difference_phi} and Lemma \ref{lemma_achievable_price}, we can show that when running Algorithm \ref{alg_RBP} with pseudo-cost function $ \bm{\phi} $ that follows Definition \ref{def_price}, there is  
\begin{align}\label{eq_diff_phi}
     \phi_{t}(x^*_{t}) -  \phi_{t-1}(x^*_{t-1}) = \frac{\alpha}{k}  \left(p_t x^*_t - p_{\min} x^*_t  \right).
\end{align}
Since in this section we do not consider the effect of $ b $, $ \bar{x}_t = x_t^* $ unless $ t = T $. We emphasize that this property is important for the proof of Theorem \ref{thm:oc_known_notice} via the \textit{online primal-dual}(OPD) approach.

\subsection{Proof of Upper Bound via OPD}\label{proof_upper}
Recall that based on the dual problem \eqref{dual_OC_max}, the dual variable $ \mu_t $ is redundant when there is no box constraint. Thus we start by considering a feasible solution $ \varphi_T $ to the dual problem \eqref{dual_OC_max}: 
\begin{align}\label{eq_lambda_hat}
   \varphi_T =  \phi_{T}(x^*_{T}), 
\end{align}
where $ \phi_{T}(x^*_{T}) = p_{\min} + \frac{\alpha p_{\min} - p_{\min}}{\prod_{i=1}^{T}(1 - \frac{\alpha}{k} x_i^*)}$. 

It is easy to show 
\begin{align}
     \phi_{T}(x^*_{t}) \ge  \phi_{t}(x^*_{t}) \ge p_t,
\end{align}
where we use the property that $  \phi_{t}(x^*_{t}) $ is monotonically increasing in $ t $. 

Based on the feasible design of $ \varphi_T $ given by Eq. \eqref{eq_lambda_hat}, we define $ \varphi_t = \phi_{t}(x^*_{t}) $ and prove Theorem \ref{thm:oc_known_notice} based on the celebrated OPD approach in the following two steps. 

\paragraph{\bfseries Step 1: Analysis of Primal-Dual Objectives}
By $ \PRM_{\boldsymbol{\phi}} $, we denote the primal and dual objectives at each step $ t \in [T] $ by $ P_t = \sum_{i=1}^t p_ix^*_i $ and $ D_t =  k \phi_{t}(x^*_{t}) $, respectively. It is easy to see that the performance of our online algorithm $ \PRM_{\boldsymbol{\phi}} $ can be represented as follows: 
\begin{align*}
    \PRM_{\boldsymbol{\phi}}(\sigma|\alpha) =    P_T = \sum_{t=1}^T p_t x^*_t.
\end{align*}

At any step $ t \in [T-1] $, the change of the primal objective under  $ \PRM_{\boldsymbol{\phi}} $ is:
\begin{align*}
    \Delta_{\textsf{Primal}}[t] = P_t - P_{t-1} = p_t x_t^*.
\end{align*}
At the last step $ T $, the change of the primal objective is:
\begin{align*}
    \hat{\Delta}_{\textsf{Primal}}[T] = p_T \bar{x}_T = p_T x_T^* + \left(k - \sum_{t=1}^T x^*_t\right)p_T,
\end{align*}
where we denote by $ \bar{x}_T $ the real quantity converted in the last step $ T $, i.e., $ \bar{x}_T = k - \sum_{t=1}^{T-1} x^*_t $. Moreover, for convenience, we define $ \Delta_{\textsf{Primal}}[T] = p_T x_T^* $ to make it consistent with the formula of $ \Delta_{\textsf{Primal}}[t] $ for $ t \in [T-1] $.

At any step $ t \in [T] $, the change of the dual objective under $ \PRM_{\boldsymbol{\phi}} $ can be calculated as follows:
\begin{align*}
    \Delta_{\textsf{Dual}} [t] & = D_t - D_{t-1}  \\
    &=   k\phi_{t}(x^*_{t}) -  k \phi_{t-1}(x^*_{t-1})\\
    &= \alpha \left[p_t x^*_t - p_{\min} x^*_t  \right],
\end{align*}
where the last equality is based on Eq. \eqref{eq_diff_phi}.
Thus, the following inequality holds 
\begin{align*}
\Delta_{\textsf{Dual}} [t] = \ & \alpha \Delta_{\textsf{Primal}} [t] - \alpha p_{\min} x^*_t  \\
\leq \ & \alpha \Delta_{\textsf{Primal}} [t], \forall t\in [T-1].
\end{align*}

Recall that we have $ \PRM_{\boldsymbol{\phi}}(\sigma|\alpha) =  P_T $. Also note that $ P_0 = 0 $ always holds. Thus, we have
\begin{align}
    \PRM_{\boldsymbol{\phi}}(\sigma|\alpha) \nonumber  = & P_T \\
    =    & \sum_{t=1}^T \Big( P_t - P_{t-1} \Big) + P_0  \nonumber \\
    =    &  \sum_{t=1}^T \Delta_{\textsf{Primal}} [t] +\left(k - \sum_{t=1}^T x^*_t\right)p_T/b_T \nonumber\\
    =    & \frac{1}{\alpha} \sum_{t=1}^T \left( \Delta_{\textsf{Dual}} [t]  + \alpha p_{\min} x^*_t \right) + \left(k - \sum_{t=1}^T  x^*_t\right)p_T \nonumber\\
    =   & \frac{1}{\alpha} \sum_{t=1}^T \Delta_{\textsf{Dual}} [t] +  p_{\min} \sum_{t=1}^T x^*_t + \left(k - \sum_{t=1}^T x^*_t\right)p_T \nonumber\\
    \overset{(i)}{\geq}    & \frac{1}{\alpha}  \sum_{t=1}^T \Big( D_t - D_{t-1} \Big) +  k p_{\min}  \nonumber\\
    =    & \frac{1}{\alpha} (D_T - D_0) + k p_{\min} \nonumber\\
    \geq & \frac{1}{\alpha} \OPT(\sigma) - \frac{1}{\alpha} D_0 + k p_{\min} \nonumber \\
    \overset{(ii)}{=}    & \frac{1}{\alpha} \OPT(\sigma), \label{eq_OPD}
\end{align}
where the inequality $(i)$ uses the fact that $ p_T \geq p_{\min} $, and the equality $(ii)$ is due to the fact that $  D_0 = k\phi_0(0) = k\alpha  p_{\min} $. 

\paragraph{\bfseries Step 2: Discussion of Primal-Dual Feasibility Conditions}
To guarantee that the above inequality $ \PRM_{\boldsymbol{\phi}}(\sigma|\alpha) \geq \frac{1}{\alpha} \OPT(\sigma) $ holds for all $ \sigma $, the primal feasibility \begin{align*}
    \sum_{i=1}^t x_i^* \leq k
\end{align*}  
and dual feasibility 
\begin{align*}
    \varphi_t =  \phi_{t}(x^*_{t}) \geq p_t 
\end{align*}
must hold at each step $ t\in[T] $. Under Algorithm \ref{alg_RBP}, the dual variable $ \varphi $ is designed to be $ \varphi_{\ell} = \phi_{\ell}(x^*_\ell) $, and $ \phi_t $ is a non-decreasing function in $ t\in [\ell] $, where $ \ell \in [T] $ is the step by which the conversion has been completed, i.e., $ \sum_{t=1}^{\ell} x^*_t = k $. Thus, $ \varphi = \varphi_{\ell} \geq \varphi_t = p_t $ holds for all $ t \in [\ell-1] $, and $ \varphi = \varphi_{\ell} \geq p_{\ell} $.  Hence, by our design of $ \phi_t $ and the dual solution $ \varphi_t $, both the primal and dual feasibility are always guaranteed for steps before $ \ell $. However, for the steps during $ \ell $ and $ T $, the primal/dual feasibility needs some further analysis.

One way to guarantee both primal and dual feasibility in all steps is to ensure that when the conversion process is completed at time $ \ell \in [T] $, we have
\begin{align*}
     \phi_{\ell}(x^*_{\ell}) = p_{\min} + \frac{\alpha p_{\min} - p_{\min}}{\prod_{t=1}^{\ell}(1 - \frac{\alpha}{k} x^*_t)} \geq p_{\max}.
\end{align*}
The rationale is that $ \PRM_{\boldsymbol{\phi}} $ will stop selling any more asset once the estimated marginal price exceeds $ p_{\max} $ (recall that $ p_t $ is upper bounded by $ p_{\max} $ for all $ t \in [T] $). It is easy to see that the above inequality can be rewritten as 
\begin{align}\label{eq_product_tau}
   \prod_{i=1}^{\ell} \left(1 - \frac{\alpha}{k} x^*_i \right) \leq \frac{\alpha p_{\min} - p_{\min}}{p_{\max} - p_{\min}}.
\end{align}
Since $ \prod_{i=1}^{\ell}(1 - \frac{\alpha}{k} x^*_i) \le ( \frac{1}{\ell} \sum_{i\in[\ell]}(1-\frac{\alpha}{k} x^*_i))^\ell $ (by the geometric–arithmetic mean inequality), to ensure that Eq. \eqref{eq_product_tau} holds, a sufficient condition can be given as follows: 
\begin{align*}
    \left[ \frac{1}{\ell} \sum_{i = 1 }^{\ell} \left( 1 - \frac{\alpha}{k} x^*_i \right) \right]^\ell \le \frac{\alpha p_{\min} - p_{\min}}{p_{\max} - p_{\min}} = \frac{\alpha - 1}{\theta - 1}.
\end{align*}
Using the fact that $ \sum_{i\in[\ell]} x^*_i = k $, the above sufficient condition is equivalent to
\begin{align}\label{Upper-bound with ell}
    \alpha \ge \ell\left[ 1 - \left(\frac{\alpha - 1}{\theta - 1}\right)^{1/\ell} \right].
\end{align}
Since $ \alpha $ is increasing with $ \ell $. To get the worst-case competitive ratio, it suffices to satisfy the above inequality with $ \ell = T $. Thus, we will produce a feasible set of primal and dual solutions under $ \PRM_{\boldsymbol{\phi}} $ as long as the balance parameter $ \alpha $ satisfies the following condition:
\begin{align} \label{eq_feasibilility}
\alpha \geq   T\left[ 1 - \left(\frac{\alpha  - 1}{\theta-1} \right)^{1/T} \right].
\end{align}

Based on the above OPD analysis, we formally prove the two \CRs in Theorem \ref{thm:oc_known_notice} for \OCK and \OCN.

\paragraph{\bfseries For {\OCK} without box constraints}
Based on Eq. \eqref{eq_OPD} and Eq. \eqref{eq_feasibilility}, the competitive ratio of $ \PRM_{\boldsymbol{\phi}} $  for {\OCK}, denoted by  $ \CR_{\textsf{known}}(\alpha)  $, satisfies
\begin{align*}
    \CR_{\textsf{known}}(\alpha)  = \max_{\sigma\in \Omega_{\textsf{known}}}   \frac{\OPT(\sigma)}{\PRM_{\boldsymbol{\phi}}(\sigma|\alpha)} \leq \alpha
\end{align*}
as long as $ \alpha $ satisfies the inequality given by Eq. \eqref{eq_feasibilility}. Thus, the upper bound for \OCK follows.

\paragraph{\bfseries For \OCN without box constraints}  In this setting, $  T $ is unknown in advance and can be infinity in the worst-case. Thus, to guarantee that Eq. \eqref{eq_feasibilility} holds for all $ T $ (so that the primal/dual feasibility is always guaranteed), it suffices to design $ \alpha $ such that Eq. \eqref{eq_feasibilility} holds with $ T\to\infty $. Let us define $ Q =\frac{\alpha -1}{\theta-1}$. Considering Eq. \eqref{eq_feasibilility} with $ T\to\infty $, we have
\begin{align}\label{CR_unknown_T}
   \alpha \ge \ & \lim_{\ell \to \infty} T(1-Q^{\frac{1}{T}}) \nonumber \\
   =\ & \lim_{T \to \infty}\frac{Q^{1/T}\ln Q/T^2}{-1/T^2} \nonumber  \\
   = \ & -\ln Q,
\end{align}
which indicates that
\begin{align*}
    \alpha & \ge \ln \frac{\theta-1}{\alpha -1}.
\end{align*}
After some algebra, we can transform Eq. \eqref{CR_unknown_T}  to Eq. \eqref{eq_alpha_notice} as follows:
\begin{align}\label{eq_alpha_notice}
   \alpha \ge 1 + W\left(\frac{\theta - 1}{e}\right).
\end{align}
We can then argue similarly that the competitive ratio of $ \PRM_{\boldsymbol{\phi}} $ for \OCN, denoted by  $ \CR_{\textsf{notice}}(\alpha)  $, satisfies 
\begin{align*}
    \CR_{\textsf{notice}}(\alpha)  = \max_{\sigma\in \Omega_{\textsf{notice}}}   \frac{\OPT(\sigma)}{\PRM_{\boldsymbol{\phi}}(\sigma|\alpha)} \leq \alpha
\end{align*}
as long as $ \alpha $ satisfies the inequality given by Eq. \eqref{eq_feasibilility} with $ T \rightarrow +\infty $, leading to the condition in Eq. \eqref{eq_alpha_notice}. Thus, we complete the proof of the upper bound for \OCN.

\subsection{Proof of Lower Bound}
\label{sec:lower_bound_proof}
To prove the lower bound results, we construct a hard instance denoted by $ \hat{\sigma} = \{p_1, \ldots, p_T\} $, where the prices are assumed to follow a strictly increasing order: $ p_{\min}\le p_1 < \ldots < p_T \le p_{\max} $.

\paragraph{\bfseries For {\OCK} without box constraints} 
For any online algorithm {\ALG}, let $c_t = \sum_{i=1}^t  x_i$ denote the total resource traded after step $t$. Thus, $\{c_t\}_{t\in[T]}$ is a non-decreasing sequence. Then the online algorithm's return can be written as
\begin{align*}
    {\ALG}(\hat{\sigma}) & = p_1 c_1 + \sum_{t=2}^{T} [c_t - c_{t-1}] p_t + (k- c_T) p_T,
\end{align*}
which implies
\begin{align*}
    {\ALG}(\hat{\sigma}) & \ge c_T p_T - \sum_{t=1}^{T-1} c_t [p_{t+1} - p_{t}] + (k- c_T) p_{\min}. 
\end{align*}
Meanwhile, the optimal offline solution returns $ \OPT(\hat{\sigma}) = k p_T$. Thus, we have
\begin{align*}
    \CR_{\textsf{known}} \ge \frac{k p_T}{c_T p_T - \sum_{t=1}^{T-1} c_t [p_{t+1} - p_{t}] + (k- c_T) p_{\min}},
\end{align*}
which equivalently gives
\begin{align*}
    c_T \ge \frac{k p_T/\CR_{\textsf{known}}- p_{\min}}{p_T - p_{\min}} + \frac{1}{p_T - p_{\min}}\sum_{t=1}^{T-1} c_t (p_{t+1} - p_{t}).
\end{align*}

Based on Gronwall's inequality \cite{jones1964fundamental}, we have Eq. \eqref{eq_c_T}.

\begin{align}
c_T &\ge \frac{  k(p_T/\CR_{\textsf{known}} - p_{\min})}{p_T - p_{\min}}  + \frac{1}{p_T - p_{\min}}\cdot \sum_{t=1}^{T-1}(p_{t+1}-p_{t})\cdot\frac{  k(p_t/\CR_{\textsf{known}} - p_{\min})}{p_t - p_{\min}} \cdot \left[\prod_{j=t+1}^{T-1} \Big(1 + \frac{p_{j+1}-p_j}{p_j - p_{\min}} \Big) \right] \nonumber\\
&= \frac{  k(p_T/\CR_{\textsf{known}} - p_{\min})}{p_T - p_{\min}} + \frac{1}{p_T - p_{\min}}\cdot \sum_{t=1}^{T-1}\left( \frac{p_T - p_{\min}}{p_{t+1} - p_{\min}} \cdot (p_{t+1}-p_{t})\cdot\frac{ k(p_t/\CR_{\textsf{known}} - p_{\min})}{p_t - p_{\min}} \right) \nonumber\\
&= \frac{  k(p_T/\CR_{\textsf{known}} - p_{\min})}{p_T - p_{\min}} +  \sum_{t=1}^{T-1}\left( \frac{p_{t+1}-p_{t}}{p_{t+1} - p_{\min}}  \cdot\frac{  k(p_t/\CR_{\textsf{known}} - p_{\min})}{p_t - p_{\min}} \right) \nonumber \\
&= \frac{  k(p_T/\CR_{\textsf{known}} - p_{\min})}{p_T - p_{\min}} + \frac{k}{\CR_{\textsf{known}}}\sum_{t=1}^{T-1} \frac{p_{t+1}-p_{t}}{p_{t+1} - p_{\min}} + p_{\min}\cdot \frac{k  (1-\CR_{\textsf{known}}) }{\CR_{\textsf{known}}}\sum_{t=1}^{T-1}\left(\frac{1}{p_t - p_{\min}} -\frac{1}{p_{t+1} - p_{\min}}\right) \nonumber\\
&= \frac{k}{\CR_{\textsf{known}}}-\frac{ k(\CR_{\textsf{known}}-1)}{\CR_{\textsf{known}}}\frac{p_{\min}}{p_1-p_{\min}}+ \frac{k}{\CR_{\textsf{known}}}\sum_{t=1}^{T-1} \frac{p_{t+1}-p_{t}}{p_{t+1} - p_{\min}}. \label{eq_c_T}
\end{align}

Since $c_T \le k$, we have
\begin{align}\label{eq_f}
    \CR_{\textsf{known}} &\ge 1+\frac{p_1-p_{\min}}{p_1}\sum_{t=1}^{T-1} \left ( \frac{p_{t+1}-p_t}{p_{t+1}-p_{\min}}\right ) \nonumber \\
    &=1+\frac{p_1-p_{\min}}{p_1}\sum_{t=2}^{T} \left ( \frac{p_{t}-p_{t-1}}{p_{t}-p_{\min}}\right ).
\end{align}
To derive a tight lower bound for $ \CR_{\textsf{known}} $, we can follow  the similar idea as \cite{time_series_search_2001} to construct $ \hat{\sigma} $ in a way that maximizes Eq. \eqref{eq_f}, leading to the following optimal competitive ratio for {\OCK}
\begin{align}\label{eq_alpha_lower_bound_finite}
    \CR^*_{\textsf{known}} =
    T\left[ 1 - \left(\frac{\CR^*_{\textsf{known}}  - 1}{\theta-1} \right)^{1/T} \right].
\end{align}
For more details on the derivation of Eq. \eqref{eq_alpha_lower_bound_finite} based on  Eq. \eqref{eq_f}, please see \cite{time_series_search_2001}. 

\paragraph{\bfseries For \OCN without box constraints} Since in this case the adversary can arbitrarily choose the value of $ T $, meaning that in the worst-case $ T $ can be infinity. Thus, the proof for {\OCN} follows trivially by setting $ T\to\infty$ for Eq. \eqref{eq_alpha_lower_bound_finite}, leading to the following optimal competitive ratio for {\OCN}:
\begin{align}\label{eq_lower_bound_unknown_T}
    \CR^*_{\textsf{notice}} = 1 + W\left(\frac{\theta - 1}{e}\right).
\end{align}
We thus complete the proof of Theorem \ref{thm:oc_known_notice} for \OCK and \OCN without box constraints.

The remaining proofs of Theorem \ref{thm:oc_known_notice} for \OCK and \OCN with box constraints can be found in Appendix \ref{sec:proof_of_theorem_CR_box_constraint} and Appendix \ref{sec:proof_of_theorem_CR_box_constraint_notice}.

\section{$ \textbf{\OCU}$: A Review} \label{appendix_OC_unknown}
For {\OCU}, an optimal algorithm has been proposed by \cite{Sun_MultipleKnapsack_2020}. 
In what follows, we briefly illustrate that the algorithm by \cite{Sun_MultipleKnapsack_2020}  can be unified into our algorithmic framework. 

In the context of {\OCU}, a unique challenge is that the player may not have the opportunity to trade all remaining resource during the final step $T$.  On the other hand, we need to avoid the risk of executing zero trading in the situation when all coming prices are $ p_{\min} $. To deal with such extreme cases, the pseudo-cost function must be designed in a way such that $ \phi(\beta) = p_{\min}$ holds at the beginning for a certain number of steps. This idea is formally presented in Lemma \ref{lemma_unknown_CR} and Lemma \ref{lemma_CR_uknown_opt} below.

\begin{lemma}\label{lemma_unknown_CR}
If the pseudo-cost function is given by $ \hat{\bm{\phi} }= (\hat{\phi}_1, \cdots, \hat{\phi}_T) $, where $ \hat{\phi}_t $ is defined as follows:
\begin{align}\label{eq_phi_hat}
\begin{split}
    \hat{\phi}_t(\beta) = 
    \begin{cases}
        p_{\min} & \text{if } c_{t-1}\in \left[0, \frac{k}{\alpha}\right),\\
        p_{\min} e^{\alpha(\beta+c_{t-1})} & \text{otherwise},
    \end{cases}
\end{split}
\end{align}
then the competitive ratio of $ \PRM_{\hat{\bm{\phi}}} $ for {\OCU}, denoted by $ \CR_{\textsf{unknown}}(\alpha) $, satisfies $ \CR_{\textsf{unknown}}(\alpha) \leq  \alpha $ as long as the balance parameter $ \alpha \geq 1+\ln\theta $. In Eq. \eqref{eq_phi_hat}, $ c_t = \sum_{t=i}^{t}x_i $ denotes the total resource traded after step $ t $. 
\end{lemma}

\begin{lemma}\label{lemma_CR_uknown_opt}
If $ \alpha = 1 + \ln \theta $, then $ \PRM_{\hat{\bm{\phi}}} $ achieves the optimal competitive ratio $ \CR^*_{\textsf{unknown}} $ given as follows
\begin{align}
    \CR^*_{\textsf{unknown}} = 1 + \ln \theta.
\end{align}
\end{lemma}

\section{Proof of Proposition \ref{prop:CR_worst_p_min}}\label{proof_p_min}
Assume $\alpha_*$ minimizes $\CR(\alpha)$, with the worst-case instance denoted by $\sigma_*$. If a forced trading phase occurs, let the switching step be $\tau_*$. Based on Algorithm \ref{alg_RBP}, for a given $ \tau_* \in [\tau_{\min}, T] $, the worst-case \CR between $ \OPT $ and $ \PRM_{\boldsymbol{\phi}} $ can be expressed as:
\begin{align*}
 \CR^*_{\textsf{known}} &=  \min_{\alpha\in [1, \theta]} \max_{\sigma\in\Omega} \frac{\OPT(\sigma)}{\PRM_{\boldsymbol{\phi}}(\sigma|\alpha)} 
           \\& =  \frac{\OPT(\sigma_*)}{\PRM_{\boldsymbol{\phi}}(\sigma_*|\alpha_*)} \\ 
           &= {\max_{\sigma_*}}\ \frac{\sum_{t=1}^{\tau_*-1} p_tx^{\OPT}_t+\sum_{t=\tau_*}^{T} p_tx^{\OPT}_t}{\sum_{t=1}^{\tau_*-1} p_t \bar{x}_t +\sum_{t=\tau_*}^{T} p_tb}.
\end{align*}

To achieve the optimal solution $\OPT(\sigma_*)$, exactly $k$ units of resource must be traded over horizon $T$. Initially, trading $b$ units per step results in a surplus of $bT - k$ units, which must be reallocated to ensure only $k$ units are traded. To help with our following proof, we formally define this operation via the following function $h$, which reallocates the surplus optimally, withdrawing from the lowest prices first: Let $d = \left\lceil \frac{bT - k}{b} \right\rceil = T - \left\lceil \frac{k}{b} \right\rceil$ represent the steps affected by reallocation, then the function $h$ is defined as:
\begin{align*}
h(\sigma_*, bT - k) := \ &  \left( \sum_{d-1} \{ p_1, \ldots, p_T \} \right) b \\ 
& + \left( \min_{d\textsf{-th}} \{ p_1, \ldots, p_T \} \right) (bd - bT + k).
\end{align*}
Here, $\sum_{d-1}$ represents the summation of the lowest $ d-1$ 
prices (i.e., $h$ withdraws up to $b$ units at each of the first $d-1$ lowest prices in $\sigma_*$) and $\min_{d\textsf{-th}}$ denotes the $ d$-th lowest price (i.e., $h$ withdraws the remaining $bd - bT + k$ units from the $d$-th lowest price). This minimizes the impact on profit and ensures exactly $k$ units are traded as required by $\OPT(\sigma_*)$. Under our setting:
\begin{align*}
    \CR^*_{\textsf{known}} &= \max_{\sigma_*}\ \frac{\sum_{t=1}^{\tau_*-1} p_t x^{\OPT}_t + \sum_{t=\tau_*}^T p_t x^{\OPT}_t}{\sum_{t=1}^{\tau_*-1} p_t \bar{x}_t + \sum_{t=\tau_*}^T p_t b} \\
    &= \max_{\sigma_*}\ \frac{\sum_{t=1}^T p_t b - h(\sigma_*, bT - k)}{\sum_{t=1}^T p_t b - \sum_{t=1}^{\tau_*-1} p_t (b-\bar{x}_t)} \\
    &= \max_{\sigma_*}\ 1 + \frac{\sum_{t=1}^{\tau_*-1} p_t (b - \bar{x}_t) - h(\sigma_*, bT - k)}{\sum_{t=1}^{\tau_*-1} p_t \bar{x}_t + \sum_{t=\tau_*}^T p_t b}.
\end{align*}

Note that $\sum_{t=1}^{\tau_*-1} (b - \bar{x}_t) = bT - k$, therefore $\sum_{t=1}^{\tau_*-1} p_t (b - \bar{x}_t) -h(\sigma_*, bT - k) \geq 0$. To achieve the maximization, we aim to minimize both $h(\{ p_1, \ldots, p_T\}, bT - k)$ and $\sum_{t=\tau_*}^T p_t b$. It is easy to see that the maximization is satisfied when $p_t = p_{\min}$ for all $t \geq \tau_*$:
\begin{align*}
    \CR^*_{\textsf{known}} &= \min_{\alpha \in [1, \theta]} \max_{\sigma \in \Omega} \frac{\OPT(\sigma)}{\PRM_{\boldsymbol{\phi}}(\sigma|\alpha)} \\
    &= \max_{p_1, \cdots, p_{\tau_*}} \frac{\sum_{t=1}^{\tau_*-1} p_t x_t^{\OPT} + \sum_{t=\tau_*}^T p_{\min} x_t^{\OPT}}{\sum_{t=1}^{\tau_*-1} p_t \bar{x}_t + \sum_{t=\tau_*}^T p_{\min} b}.
\end{align*}

We thus complete the proof.

\section{Proof of Proposition \ref{prop:sigma_tau}}
\label{sec:proof_of_proposition_sigma_tau}
To begin with, we refer to $ \sigma^{(\tau)} $ as the worst case with switching step $ \tau $ based on Proposition \ref{prop:CR_worst_p_min}, We refer to the original LP in Eq. \eqref{eq:OC_max}  as  $ LP_T $, and the reduced linear program Eq. \eqref{eq_reduce} as $ LP_{\tau} $, where forced trading starts at step $ \tau $. Let $ x'_t $ for $ t \in [1, T] $ and $ x''_t $ for $ t \in [1, \tau] $ be the optimal solutions for $ LP_T $ and $ LP_{\tau} $, respectively. The objective values for $ LP_T $ and $ LP_{\tau} $ are:
\begin{align*}
\sum_{t=1}^T p_t x'_t = \sum_{t=1}^{\tau-1} p_t x'_t + p_{\min}\left(k - \sum_{t=1}^{\tau-1} x'_t\right)^+,
\end{align*}
and
\begin{align*}
\sum_{t=1}^{\tau} p_t x''_t = \sum_{t=1}^{\tau-1} p_t x''_t + p_{\min}\left(k - \sum_{t=1}^{\tau-1} x''_t\right)^+.
\end{align*}

If $ \left(k - \sum_{t=1}^{\tau-1} x''_t\right)^+ > 0 $, some resource is traded at the minimum price $ p_{\min} $. The optimal solution involves trading $ b $ units of resource at each step $ t \in [1, \tau-1] $ until the remaining resource is traded at $ p_{\min} $. Thus:
\begin{align*}
x''_t = b \quad \forall t \in [1, \tau-1],
\end{align*}
which means that some resource will inevitably be traded at the minimum price. The total revenue in both $ LP_T $ and $ LP_{\tau} $ becomes:
\begin{align*}
\sum_{t=1}^T p_t x'_t = \sum_{t=1}^{\tau-1} p_t b + p_{\min}\left(k - b(T - \tau)\right),
\end{align*}
which is the same for both $ LP_T $ and $ LP_{\tau} $. Therefore:
\begin{align*}
\sum_{t=1}^T p_t x'_t = \sum_{t=1}^T p_t x''_t.
\end{align*}

On the other hand, if $ \left(k - \sum_{t=1}^{\tau-1} x''_t\right)^+ = 0 $, no resource remains to be traded at $ p_{\min} $. The optimal strategy in both $ LP_T $ and $ LP_{\tau} $ is to trade resource at the highest prices within $ (p_1, \ldots, p_{\tau-1}) $, allocating $ b $ units of resource to the highest price, then proceeding to the next highest price, and so on. Since the prices $ (p_1, \ldots, p_{\tau-1}) $ are the same for both $ LP_T $ and $ LP_{\tau} $, the total revenue is identical:
\begin{align*}
   \sum_{t=1}^T p_t x'_t = \sum_{t=1}^T p_t x''_t.
\end{align*}

In both cases, the total revenue from $ LP_T $ is equal to that from $ LP_{\tau} $. Therefore, solving the original $ LP_T $ is equivalent to solving the reduced $ LP_{\tau} $ for the worst-case price sequence $ \sigma^{(\tau)} $, completing the proof.

\section{Proof of Theorem \ref{thm:oc_known_notice} for \OCK with Box Constraints}\label{sec:proof_of_theorem_CR_box_constraint}
This section provides the remaining proofs for Theorem \ref{thm:oc_known_notice} with box constraints. Similar to Appendix \ref{sec:proof_of_thm_known_notice}, we omit $ ( {F}_t,  \alpha) $ and directly write the pseudo-cost function $ \phi_t(x_t|  {F}_t,\alpha) $ as $ \phi_t(x_t) $.

\subsection{Bounding the Competitive Ratio with $ \tau$}\label{proof_bound_cr}
Building on Propositions \ref{prop:CR_worst_p_min} and \ref{prop:sigma_tau}, we develop the following lemma to facilitate the competitive analysis of \OCK with box constraints.

\begin{lemma}\label{lemma_ratio_tau}
For price sequence $\sigma^{(\tau)}$ with $\tau \in [\tau_{\min}, T]$, the following upper bound holds:
\begin{align*}
\frac{\OPT(\sigma^{(\tau)})}{\PRM_{\boldsymbol{\phi}}(\sigma^{(\tau)}|\alpha)} &\leq \alpha,
\end{align*}
where $ \alpha $ satisfies the following equation
\begin{align*}
 \alpha  = \tau \left[ 1 - \left( \frac{\alpha - 1}{\theta - 1} \right)^{1/\tau} \right].
\end{align*}
\end{lemma}

\begin{proof}
Based on Proposition \ref{prop:CR_worst_p_min} and Proposition \ref{prop:sigma_tau}, $\OPT(\sigma^{(\tau)})$  can be obtained by solving Eq. \eqref{eq_reduce} with $ \sigma = \sigma_r^{(\tau)}$.

We again use the online primal-dual (OPD) approach. Recall that the dual problem is given by:
\begin{subequations}
\begin{align}
   \min_{ \varphi, \bm{\mu}}\quad & k\varphi+\sum_{t=1}^{\tau} b_t\mu_t \\
    s.t. \quad & \varphi+\mu_t \ge p_t, & & \forall t\in[{\tau}],\\ 
                & \varphi, \mu_t \ge 0,  & & \forall t\in[{\tau}].
\end{align}
\end{subequations}
At any step $ t \in [\tau] $, the change of the primal objective is:
\begin{align*}
    \Delta_{\textsf{Primal}}[t]= P_t - P_{t-1} = p_t  \bar{x}_t, 
\end{align*}
where we denote by $ \bar{x}_t $ the {real} total quantity traded at step $ t $. 
\begin{align*}
    \bar{x}_t=\begin{cases}
        x_t^* & t<\tau,\\
        k-\sum_{t=1}^{\tau-1}x_t^* & t=\tau.
    \end{cases}
\end{align*}

Based on our design of Algorithm \ref{alg_RBP}, we know that $ \bar{x}_{\tau} > x_{\tau}^* $ and $ p_{\tau} =\phi_{\tau}( {x}_{\tau}^*) < \phi_{\tau}(\bar{ {x}}_{\tau})  $.
Follow the OPD approach, we design $ \varphi_t = \phi_t(  \bar{x}_t) $ and $ \mu_t  = \max\{0, p_t - \phi_t(  \bar{x}_t \}  $. It is easy to verify that  $ \varphi_t  $ and $ \mu_t  $ are always dual feasible.

At any step $ t \in [\tau-1] $, we have $ \bar{x}_t = x_t^* $ and $ b_t = b$, and thus the change of the dual objective can be calculated as follows:
\begin{align*}
    \Delta_{\textsf{Dual}} [t]  =\ &  D_t - D_{t-1}  
   \\ =\ & k \left( \phi_t( {x}_t^*) - \phi_{t-1}( {x}_{t-1}^*) \right) + b\mu_t   \\ 
    =\ & k \left( \phi_t( {x}_t^*) - \phi_{t-1}( {x}_{t-1}^*) \right) +  b\max\{0, p_t - \phi_t( {x}_t^*) \} \\
    =\ & \alpha \left( \phi_t( {x}_t^*) x_t^* -  p_{\min} x_t^* \right) + b\max\{0, p_t - \phi_t( {x}_t^*) \}\\
    \leq\ &  \alpha p_tx_t^* - \alpha p_{\min} x_t^*,
\end{align*}
where the last inequality can be proved in the following two cases:   
\begin{itemize}
    \item When $ \max\{0, p_t - \phi_t( {x}_t^*) \} = p_t - \phi_t( {x}_t^*) $, we have $  \phi_t( {x}_t^*) \le p_t $. Thus, 
    \begin{align*}
       & \alpha \left( \phi_t( {x}_t^*) x_t^* -  p_{\min} x_t^* \right) + b\max\{0, p_t - \phi_t( {x}_t^*) \} \\ 
       =\ &  
       \alpha \left( \phi_t( {x}_t^*) x_t^* -  p_{\min} x_t^* \right) +  bp_t  - b\phi_t( {x}_t^*) \\
       =\ & \alpha  \phi_t( {x}_t^*) x_t^* - b\phi_t( {x}_t^*)  +  bp_t  - \alpha p_{\min} x_t^*   \\
       \le\ & 
       \alpha p_tx_t^* - \alpha p_{\min} x_t^*,  
   \end{align*}
   where the equality holds if $ p_t= \phi_t( {x}_t^*) $ .
    \item When $ \max\{0, p_t - \phi_t( {x}_t^*) \} = 0 $, we have $ p_t < \phi_t( {x}_t^*) $ and $ x_t^* = 0 $. Thus,
   \begin{align*}
      & \alpha \left( \phi_t( {x}_t^*) x_t^* -  p_{\min} x_t^* \right) + b\max\{0, p_t - \phi_t( {x}_t^*) \} \\=& \alpha p_tx_t^* - \alpha p_{\min} x_t^* \\ = & 0.
   \end{align*}
\end{itemize}

Combining the above two cases, we can derive the following inequality:
\begin{align*}
\Delta_{\textsf{Dual}} [t] \leq\ & \alpha \Delta_{\textsf{Primal}} [t] - \alpha p_{\min} x^*_t  \\
\leq \ & \alpha \Delta_{\textsf{Primal}} [t], \qquad \forall t \in [\tau-1].
\end{align*}

Then, with $ P_0 = 0 $, we can follow the online primal-dual approach and obtain the following inequality
\begin{align*}
    \PRM_{\boldsymbol{\phi}}(\sigma^{(\tau)}|\alpha) =  &P_{\tau}
   \\ =    & \sum_{t=1}^{\tau} \Big( P_t - P_{t-1} \Big) + P_0   \\
    =    &  \sum_{t=1}^{\tau} \Delta_{\textsf{Primal}} [t]  \\
    = & \sum_{t=1}^{\tau} p_tx_t^*+p_{\min}  (k-\sum_{t=1}^{\tau}x_t^*)  \\
    \overset{(i)}{\ge}    
   & \frac{1}{\alpha} \sum_{t=1}^{\tau} \left( \Delta_{\textsf{Dual}} [t]  + \alpha p_{\min} x_t^* \right) +p_{\min}  (k-\sum_{t=1}^{\tau}x_t^*)\\
    =   & \frac{1}{\alpha} \sum_{t=1}^{\tau} \Delta_{\textsf{Dual}} [t]  + k p_{\min}\\
    \ge    & \frac{1}{\alpha}  \sum_{t=1}^{\tau} \Big( D_t - D_{t-1} \Big) + k p_{\min}  \\
    =    & \frac{1}{\alpha} (D_{\tau} - D_0) + k p_{\min}\\
    \overset{(ii)}{=}     & \frac{1}{\alpha} D_{\tau}\\
    {\ge} &\frac{1}{\alpha} \OPT(\sigma^{(\tau)})\\
    =&\frac{1}{\alpha} \OPT(\sigma_r^{(\tau)}),
\end{align*}
where the inequality $(i)$ uses the incremental inequality and the equality $(ii)$ is due to the fact that $  D_0 = k \phi_0(0) + b\mu_0 = \alpha k p_{\min} $ .

To guarantee that 
\begin{align*}
     \PRM_{\boldsymbol{\phi}}(\sigma^{(\tau)}|\alpha) \geq \frac{1}{\alpha} \OPT(\sigma_r^{(\tau)})
\end{align*}
holds for all $ \sigma_r^{(\tau)} $, the primal and dual solutions must be feasible at each step $ t \in [\tau] $ so that weak duality can be applied to connect $ D_{\tau} $ with $ \OPT(\sigma_r^{(\tau)}) $:  
\begin{align*}
    D_{\tau} \geq D_* \geq P_* = \OPT(\sigma_r^{(\tau)}), 
\end{align*}
where $ \OPT(\sigma_r^{(\tau)}) $ is the offline optimal objective value of the reduced LP problem with input instance $ \sigma_r^{(\tau)} $. So in the following, we prove the conditions for ensuring primal and dual feasibility of the online solution.

At each step $ t \in [\tau] $, a feasible primal solution must satisfy the following condition:
\begin{align}\label{eq_primal_feasibility_condition_tau}
    \sum_{i=1}^t \bar{x}_t \leq k,
\end{align}
and we must guarantee the following condition for a feasible dual solution:
\begin{align}\label{eq_dual_feasibility_condition_tau}
    \varphi_t + \mu_t = \phi_t( \bar{x}_t) + \max \{ 0, p_t - \phi_t( \bar{x}_t) \} \geq p_t. \qquad
\end{align}
It is easy to see that the dual feasibility above always holds based on the design of $ \mu_t $. To ensure primal feasibility at each step, following the same steps as we did in Appendix \ref{proof_upper}, a sufficient condition to ensure primal  feasibility is:
\begin{align*}
     \phi_{\tau}( \bar{x}_{\tau}) = p_{\min} + \frac{\alpha p_{\min} - p_{\min}}{\prod_{i=1}^{\tau} \left( 1 - \frac{\alpha}{k} \bar{x}_i \right)} \geq p_{\max}.
\end{align*}
Using the fact that $ \sum_{i=1}^{\tau} \bar{x}_i = k $, the above sufficient condition is equivalent to:
\begin{align*}
    \alpha & \geq \tau \left[ 1 - \left( \frac{\alpha p_{\min} - p_{\min}}{p_{\max} - p_{\min}} \right)^{1/\tau} \right] \\
    & = \tau \left[ 1 - \left( \frac{\alpha - 1}{\theta - 1} \right)^{1/\tau} \right].
\end{align*}
We thus complete the proof of Lemma \ref{lemma_ratio_tau}.
\end{proof}

\subsection{Competitive Ratio with  $\alpha \in [\alpha_{\tau}, \alpha_{\tau+1}) $}
If we define $ \alpha_{\tau}  $ to be the unique solution to the following:
\begin{align}\label{eq_alpha_tau_appendix}
    \alpha_{\tau} = \tau \left[ 1 - \left( \frac{\alpha_{\tau} - 1}{\theta - 1} \right)^{1/\tau} \right], \qquad \tau \in [\tau_{\min}, T].
\end{align}
The following lemma helps establish the relationship between $\CR(\alpha)$ and $\alpha_{\tau} $:
\begin{lemma}\label{lemma_cr_tau} 
For any $ \alpha \in [\alpha_{\tau}, \alpha_{\tau+1}) $ with $ \tau \in [\tau_{\min}, T] $, the following upper bound holds: 
\begin{align*} 
\CR(\alpha) &= \max_{\sigma \in \Omega} \frac{\OPT(\sigma)}{\PRM_{\boldsymbol{\phi}}(\sigma|\alpha)}\le \alpha. 
\end{align*} 
\end{lemma}

\begin{proof}
The proof is structured in the following two parts:

For $ \hat{\tau} \in [\tau_{\min}, \tau]$: following Lemma \ref{lemma_ratio_tau}, there is 
\begin{align*}
   \frac{\OPT(\sigma^{(\hat{\tau})})}{\PRM_{\boldsymbol{\phi}}(\sigma^{(\hat{\tau})}|\alpha)} \le \alpha_{\hat{\tau}},
\end{align*}
and since $\alpha \ge \alpha_{\tau} \ge \alpha_{\hat{\tau}}$,  the following inequality holds:
\begin{align}\label{eq_lower_tau}
    \frac{\OPT(\sigma^{(\hat{\tau})})}{\PRM_{\boldsymbol{\phi}}(\sigma^{(\hat{\tau})}|\alpha)} \le \alpha, \quad \forall \alpha \in [\alpha_{\tau}, \alpha_{\tau+1}), \hat{\tau} \in  [\tau_{\min}, \tau].
\end{align}

For $\hat{\tau} \in (\tau, T]$: We start by proving the upper bound under $\sigma^{(\tau+1)}$, to further simplify the notation, we use $ \phi_{\tau} $ to represent $ \phi_{\tau}( \bar{x}_{\tau})$, with given $\sigma^{(\tau+1)}$, we have $ p_{\tau+1}=p_{\min}<\phi_{\tau}$, and therefore $\phi_{\tau+1}= \phi_{\tau} $, $ \mu_{\tau+1}=0$, and 
\begin{align}\label{eq:OPT_sigma_tau_1}
    \OPT(\sigma^{(\tau+1)}) = k\phi_{\tau+1}+b\sum_{t=1}^{\tau+1}\mu_t  = k\phi_{\tau}+b\sum_{t=1}^{\tau}\mu_t, 
\end{align}
where $ \mu_{\tau}=\max\{0, p_{\tau}-\phi_{\tau}\}$ and the second equality is due to the fact that $\phi_{\tau+1}= \phi_{\tau} $ and $ \mu_{\tau}+1=0 $.

Then, we can show that the following holds:
\begin{align*}
&\frac{\OPT(\sigma^{(\tau+1)})}{\PRM_{\boldsymbol{\phi}}(\sigma^{(\tau+1)}|\alpha)} \\
= \ &  \frac{k \phi_{\tau}+b\sum_{t=1}^{\tau}\mu_t}{\sum_{t=1}^{\tau}p_t \bar{x}_t + p_{\min}(k - \sum_{t=1}^{\tau} p_t \bar{x}_t)} \\ 
= \ &  \frac{k \phi_{\tau-1} + k(\phi_{\tau} - \phi_{\tau-1})+b\sum_{t=1}^{\tau-1}\mu_t+b\mu_{\tau}}{\sum_{t=1}^{\tau-1}p_t \bar{x}_t + p_{\min}(k - \sum_{t=1}^{\tau-1}p_t \bar{x}_t) + (p_{\tau} - p_{\min})\bar{x}_{\tau}} \\
\overset{(i)}{=} \ &  \frac{\OPT(\sigma^{(\tau)}) + k(\phi_{\tau} - \phi_{\tau-1})+b\mu_{\tau}}{\PRM_{\boldsymbol{\phi}}(\sigma^{(\tau)}|\alpha) + (p_{\tau} - p_{\min})\bar{x}_{\tau}} \\
\overset{(ii)}{\leq}\ &  \frac{\OPT(\sigma^{(\tau)}) + k(\phi_{\tau} - \phi_{\tau-1})+b\mu_{\tau}}{\frac{\OPT(\sigma^{(\tau)})}{\alpha} + (p_{\tau} - p_{\min})\bar{x}_{\tau}} \\
=\ &  \alpha \frac{\OPT(\sigma^{(\tau)}) + k(\phi_{\tau} - \phi_{\tau-1})+b\mu_{\tau}}{\OPT(\sigma^{(\tau)}) + \alpha \bar{x}_{\tau}(p_{\tau} - p_{\min})} \\
\leq \ &  \alpha,
\end{align*}
where the equality $(i)$ is due to the fact that $ \OPT(\sigma^{(\tau)}) = k \phi_{\tau-1} + b\sum_{t=1}^{\tau-1}\mu_t $, similar to Eq. \eqref{eq:OPT_sigma_tau_1}, and the inequality $(ii)$ is based on Eq. \eqref{eq_lower_tau}. The last inequality can be proved if the following inequality holds:
\begin{align}\label{eq:alpha_tau_phi_tau}
    \alpha \bar{x}_{\tau}(p_{\tau} - p_{\min}) \ge k(\phi_{\tau} - \phi_{\tau-1})+b\mu_{\tau}.
\end{align}
Thus, in the following we prove the above inequality \eqref{eq:alpha_tau_phi_tau} based on different cases of $ \bar{x}_{\tau} $.
\begin{itemize}
    \item When $\bar{x}_{\tau} = 0$, we have $ p_{\tau}<\phi_{\tau}$, so $ \phi_{\tau} = \phi_{\tau-1} $ and $ \mu_{\tau} = 0 $, therefore
    \begin{align*}
    \alpha \bar{x}_{\tau} (p_{\tau} - p_{\min}) = k(\phi_{\tau} - \phi_{\tau-1})+b\mu_{\tau} =0.
    \end{align*}
    \item When $\bar{x}_{\tau} \in (0, b)$, there are $ p_{\tau}=\phi_{\tau}$ and $ \mu_{\tau} = 0$, based on Eq. \eqref{eq_alg_x_t}, we have
    \begin{align*}
        \bar{x}_{\tau}=\frac{k}{\alpha}\frac{p_{\tau}-\phi_{\tau-1}}{p_{\tau}-p_{\min}},
    \end{align*}
    which can be used to prove the following equation
    \begin{align*}
    \alpha \bar{x}_{\tau}(p_{\tau} - p_{\min}) 
    & = \alpha \bar{x}_{\tau}(p_{\tau} - p_{\min}) 
    \\& =k(p_{\tau} - \phi_{\tau-1})
    \\&= k(\phi_{\tau} - \phi_{\tau-1})
     \\&= k(\phi_{\tau} - \phi_{\tau-1})+ b\mu_{\tau}.
    \end{align*}
    \item When $\bar{x}_{\tau} = b$, we have $ p_{\tau}>\phi_{\tau}$ and $ \mu_{\tau} = p_{\tau}-\phi_{\tau}$, still based on Eq. \eqref{eq_alg_x_t}:
        \begin{align*}
            \frac{k(\phi_{\tau}-\phi_{\tau-1})}{\alpha(p_{\tau}-p_{\min})} \le \frac{k(\phi_{\tau}-\phi_{\tau-1})}{\alpha(\phi_{\tau}-p_{\min})} = b,
        \end{align*}
        which implies that
        \begin{align*}
            \alpha\bar{x}_{\tau}(p_{\tau}-p_{\min})&=\alpha b(p_{\tau}-p_{\min})\\
            &=\alpha b(\phi_{\tau}-p_{\min})+\alpha b(p_{\tau}-\phi_{\tau})\\
            & = k(\phi_{\tau}-\phi_{\tau-1})+\alpha b(p_{\tau}-\phi_{\tau})
            \\&\ge k(\phi_{\tau}-\phi_{\tau-1})+b(p_{\tau}-\phi_{\tau})
            \\&= k(\phi_{\tau}-\phi_{\tau-1})+b\mu_{\tau}.
        \end{align*}
\end{itemize}

We can repeat the same procedure for $\sigma^{(\hat{\tau})}$ with $ \hat{\tau} = \tau+2, \cdots, T $, and draw the following conclusion: 
\begin{align}\label{eq_upper_tau}
\frac{\OPT(\sigma^{(\hat{\tau})})}{\PRM_{\boldsymbol{\phi}}(\sigma^{(\hat{\tau})}|\alpha)} \le \alpha, \quad \forall \alpha \in [\alpha_{\tau}, \alpha_{\tau+1}), \hat{\tau} \in (\tau, T].
\end{align}

Recall that in Eq. \eqref{omega_tau}, we define $\Omega^{(\tau)}$ as the set of worst-case instances where the switching occurs exactly at step $\tau$. Formally, $\Omega^{(\tau)}$ is defined as follows:
\begin{align*} 
\Omega^{(\tau)} :=\ & \Big\{ \sigma = \big(p_1, p_2, \ldots, p_{\tau}, p_{\min}, \\
&  \qquad \ldots, p_{\min} \big) \big|\ p_t \in [p_{\min},p_{\max}], \forall t \in [1, \tau] \Big\},
\end{align*} 
where $\Omega^{(\tau)}$ represents all the input sequences in which prices remain between $ [p_{\min}, p_{\max}]$  until step $\tau$, and then switch to $p_{\min}$ for all $t \geq \tau$. Therefore, to compute the worst-case competitive ratio, we just need to consider the following set of input instances: 
\begin{align*} 
\Omega^{(\tau_{\min})} \cup \Omega^{(\tau_{\min} + 1)} \cup \cdots \cup \Omega^{(T)}, 
\end{align*} 
where each subset $\Omega^{(\tau)}$ corresponds to the worst-case instances with a switching step occurring exactly at $\tau$.

Combining Eq. \eqref{eq_lower_tau} and Eq. \eqref{eq_upper_tau}, we have
\begin{align*}
\CR(\alpha) &= \max_{\sigma \in \Omega} \frac{\OPT(\sigma)}{\PRM_{\boldsymbol{\phi}}(\sigma|\alpha)} \\
&= \max_{\sigma \in \Omega^{(\tau_{\min})} \cup \cdots \cup \Omega^{(T)}} \frac{\OPT(\sigma)}{\PRM_{\boldsymbol{\phi}}(\sigma|\alpha)} \\
& \le \alpha, \qquad \forall \alpha \in [\alpha_{\tau}, \alpha_{\tau+1}), \tau \in  [\tau_{\min}, T].
\end{align*}
We thus complete the proof of Lemma \ref{lemma_cr_tau}.
\end{proof}

\subsection{Putting Everything Together: Proof of Theorem \ref{thm:oc_known_notice} for \OCK with Box Constraints}
To prove the Theorem \ref{thm:oc_known_notice}, we combine the results from Lemma \ref{lemma_ratio_tau} and Lemma \ref{lemma_cr_tau}:
\begin{itemize}
    \item Lemma \ref{lemma_ratio_tau} shows that, the competitive ratio under $ \sigma^{(\tau)} $ is upper bounded by $ \alpha_{\tau} $, giving $ \tau \in [\tau_{\min}, T]$.
    \begin{align*} 
    \frac{\OPT(\sigma^{(\tau)})}{\PRM_{\boldsymbol{\phi}}(\sigma^{(\tau)}|\alpha)} \le \alpha_{\tau}, \quad \forall \sigma^{\tau} \in \Omega^{(\tau)}, \tau \in [\tau_{\min}, T],
    \end{align*} 
    where $ \alpha_{\tau} $ is defined by Eq. \eqref{eq_alpha_tau_appendix}.

    \item Lemma \ref{lemma_cr_tau} demonstrates that for all $\tau \in [\tau_{\min}, T]$, we have:
    \begin{align*}
    \CR(\alpha) = \max_{\sigma \in \Omega} \frac{\OPT(\sigma)}{\PRM_{\boldsymbol{\phi}}(\sigma|\alpha)} \le \alpha, \forall \alpha \in [\alpha_{\tau}, \alpha_{\tau+1}).
    \end{align*}

    \item As $\alpha_{\tau}$ increases with $\tau$, we conclude:
    \begin{align*}
    \CR^*_{\textsf{known}} = \min_{\alpha \in [\alpha_{\tau_{\min}}, \alpha_T]} \max_{\sigma \in \Omega} \frac{\OPT(\sigma)}{\PRM_{\boldsymbol{\phi}}(\sigma|\alpha)} \le \alpha_{\tau_{\min}}. 
    \end{align*}
    
\end{itemize}
Substituting $\tau_{\min} = T - \lceil \frac{k}{b} \rceil + 1$ into $\alpha_{\tau_{\min}}$ by Eq. \eqref{eq_alpha_tau_appendix}, we obtain 
\begin{align}\label{eq_alpha_tau_min}
      \alpha_{\tau_{\min}} = \left(T - \lceil \frac{k}{b} \rceil + 1\right) \left[ 1 - \left(\frac{ \alpha_{\tau_{\min}} - 1}{\theta - 1} \right)^{\frac{1}{T - \lceil \frac{k}{b} \rceil + 1}} \right].
\end{align} 
We thus complete the proof of Theorem \ref{thm:oc_known_notice} for \OCK with non-trivial box constraints.

\section{Proof of Theorem \ref{thm:oc_known_notice} for \OCN with Box Constraints}
\label{sec:proof_of_theorem_CR_box_constraint_notice}

For \OCN with box constraints, it is well established in previous research, such as \cite{lechowicz2024online}, that the competitive ratio $ \CR^*_{\textsf{notice}} $ is optimal. While the proofs for the upper and lower bounds are already known, our work presents a new algorithmic approach for achieving this optimal competitive ratio. For this reason, here we only provide our new proof for the upper bound, connecting it directly to the results in Appendix \ref{sec:proof_of_theorem_CR_box_constraint}.

To prove the upper bound for \OCN with a box constraint, we apply the same methodology used in Appendix \ref{proof_upper}, where we designed the parameter $ \alpha $ to satisfy the condition in Eq. \eqref{eq_alpha_tau_min}, taking the limit as $ T \to \infty $. 

Since $\tau_{\min} = T - \lceil \frac{k}{b} \rceil + 1 = T $ as $ T \to \infty $, Eq. \eqref{eq_alpha_tau_min} can be transformed into the same form as Eq. \eqref{eq_feasibilility}.

Following the steps outlined in Appendix \ref{proof_upper}, and using Eq. \eqref{eq_feasibilility}, we find that the condition for $\alpha$ remains the same:
\begin{align*}
   \alpha \ge 1 + W\left(\frac{\theta - 1}{e}\right),
\end{align*}
where $W$ denotes the Lambert-W function. Therefore, the following upper bound holds
\begin{align*}
    \CR_{\textsf{notice}}(\alpha)  = \max_{\sigma\in \Omega_{\textsf{notice}}}   \frac{\OPT(\sigma)}{\PRM_{\boldsymbol{\phi}}(\sigma|\alpha)} \leq \alpha
\end{align*}
as long as $\alpha$ satisfies the inequality given by Eq. \eqref{eq_feasibilility} with $ T \rightarrow \infty $. This leads to the same condition as Eq. \eqref{eq_alpha_notice}, thus completing the proof of the upper bound for \OCN with box constraints.

\section{Proof of Theorem \ref{thm:combine_cr}}\label{appendix_la}
To describe the results of running Algorithm \ref{alg_RBP_combine} under the \OCP setting with the predicted horizon $T_{\textsf{pred}}$, we introduce the following definitions. Given a price sequence $\sigma$ and total available resource $k$, we split $k$ into two portions: $k^{(1)} = (1 - \lambda)k$ and $k^{(2)} = \lambda k$, where $\lambda$ is a hyperparameter representing the confidence level of the prediction.

\begin{itemize}
\item Performance under \OCK with $T_{\textsf{pred}}$: The performance of the algorithm when allocating $k^{(1)}$ units of resource, treating the problem as \OCK with the predicted horizon $T_{\textsf{pred}}$: 
\begin{align*} 
\PRM_{\boldsymbol{\phi_1}}(\sigma|\alpha_1) = \sum_{t=1}^T p_t \bar{x}_t^{(1)}, 
\end{align*} 
where $\bar{x}_t^{(1)}$ represents the resource traded at time step $t$ for $k^{(1)}$ based on Algorithm \ref{alg_RBP_combine}.

\item Performance under \OCU: The performance of the algorithm when allocating $k^{(2)}$ units of resource, treating the problem as if it were \OCU: 
\begin{align*} 
\PRM_{\boldsymbol{\phi_1}}(\sigma|\alpha_2) = \sum_{t=1}^T p_t x_t^{(2)}, 
\end{align*} 
where $ x_t^{(2)}$ represents the resource traded at time step $t$ for $k^{(2)}$ based on Algorithm \ref{alg_RBP_combine}.

\item Optimal offline results for $k^{(1)}$: The optimal offline result for the allocation of $k^{(1)}$ units of resource: 
\begin{align*} 
\OPT_1(\sigma) = \frac{k^{(1)}}{k} \OPT(\sigma), \end{align*} 
representing the optimal profit obtained from trading $k^{(1)}$ units under the price sequence $\sigma$.

\item Optimal offline results for $k^{(2)}$: The optimal offline result for the trading of $k^{(2)}$ units of resource: 
\begin{align*} 
\OPT_2(\sigma) = \frac{k^{(2)}}{k} \OPT(\sigma), \end{align*} 
representing the optimal profit obtained from trading $k^{(2)}$ units under the price sequence $\sigma$.
\end{itemize}

These definitions allow us to analyze the performance of Algorithm \ref{alg_RBP_combine} by separately considering the contributions from both the \OCK and \OCU components under the predicted horizon $T_{\textsf{pred}}$.

\paragraph{\bfseries Proof of Consistency} 
To establish consistency, let $T_{\textsf{pred}} = T$. According to Theorem \ref{thm:oc_known_notice}, we have $\PRM_{\boldsymbol{\phi_1}}(\sigma|\alpha_1) \geq \frac{1}{\alpha_1} \OPT_1(\sigma)$, where both the primal and dual feasibility hold when 
\begin{align*}
\alpha_1 \geq T_{\textsf{pred}} \left[ 1 - \left(\frac{\alpha_1 - 1}{\theta - 1}\right)^{1/T_{\textsf{pred}}} \right].
\end{align*} 
Similarly, the condition $\alpha_2 \geq 1 + \ln \theta$ ensures a feasible online primal-dual analysis under \OCU, leading to $\PRM_{\boldsymbol{\phi_2}}(\sigma|\alpha_2) \geq \frac{1}{\alpha_2} \OPT_2(\sigma)$ \cite{Sun_MultipleKnapsack_2020}. 

Therefore, we have
\begin{align*}
    \eta(\lambda) &= \frac{\OPT(\sigma)}{\ALG(\sigma)}\\
    &= \frac{\OPT_1(\sigma) + \OPT_2(\sigma)}{\PRM_{\boldsymbol{\phi_1}}(\sigma|\alpha_1) + \PRM_{\boldsymbol{\phi_2}}(\sigma|\alpha_2)}\\
    &\le \frac{k^{(1)}\OPT(\sigma) + k^{(2)}\OPT(\sigma)}{\frac{k^{(1)}}{\alpha_1} \OPT(\sigma) + \frac{k^{(2)}}{\alpha_2} \OPT(\sigma)}\\
    &= \frac{\alpha_1 \alpha_2}{\alpha_2 + \lambda(\alpha_1 - \alpha_2)},
\end{align*}
where $\alpha_1 = T_{\textsf{pred}}[ 1 - (\frac{\alpha_1 - 1}{\theta - 1} )^{1/T_{\textsf{pred}}} ]$, $T_{\textsf{pred}} = T$, and $\alpha_2 = 1 + \ln \theta$.

\paragraph{\bfseries Proof of Robustness} 
We analyze the robustness of the algorithm based on the relationship between $T_{\textsf{pred}}$ and $T$.

When $T > T_{\textsf{pred}}$, $\PRM_{\boldsymbol{\phi_1}}(\sigma|\alpha_1)$ completes trading $k^{(1)}$ units of resource at step $T_{\textsf{pred}}$, before reaching the final step $T$. As we did for the proof of consistency, the condition $\alpha_2 \geq 1 + \ln \theta$ ensures a feasible online primal-dual analysis under \OCU, leading to $\PRM_{\boldsymbol{\phi_2}}(\sigma|\alpha_2) \geq \frac{1}{\alpha_2} \OPT_2(\sigma)$ \cite{Sun_MultipleKnapsack_2020}.

The worst-case scenario arises when we set $p_t = p_{\min}$ for all $t \in [1, T_{\textsf{pred}}]$ within $\sigma$ since this will minimize $ \PRM_{\boldsymbol{\phi_1}}(\sigma|\alpha_1) $ without decreasing $\OPT(\sigma)$. Therefore, we have the following bound:
\begin{align*}
    \gamma(\lambda) &= \frac{\OPT(\sigma)}{\ALG(\sigma)}\\
    &= \frac{\OPT_1(\sigma) + \OPT_2(\sigma)}{\PRM_{\boldsymbol{\phi_1}}(\sigma|\alpha_1) + \PRM_{\boldsymbol{\phi_2}}(\sigma|\alpha_2)}\\
    &\leq \frac{\frac{k^{(1)}}{k} \OPT(\sigma) + \frac{k^{(2)}}{k} \OPT(\sigma)}{\PRM_{\boldsymbol{\phi_1}}(\sigma|\alpha_1) + \frac{k^{(2)}}{\alpha_2 k}\OPT(\sigma)}\\
    &\leq \frac{\frac{k^{(1)}}{k} \OPT(\sigma) + \frac{k^{(2)}}{k} \OPT(\sigma)}{k^{(1)} p_{\min} + \frac{k^{(2)}}{\alpha_2k} \OPT(\sigma)}\\
    &= \frac{\theta}{1 - \lambda + \frac{\lambda}{\alpha_2} \theta}.
\end{align*}
The last inequality holds because $\OPT(\sigma) \leq k p_{\max}$.

When $T < T_{\textsf{pred}}$, there is no forced trading phase, and the following two cases arise:

\paragraph{Case-I: No usage of $\PRM_{\boldsymbol{\phi_1}}$}
If $\max(\sigma) \leq \alpha_1 p_{\min}$, then $\PRM_{\boldsymbol{\phi_1}}(\sigma|\alpha_1)$ is minimized to $0$, resulting in
\begin{align*}
    \gamma(\lambda) \leq \frac{\OPT(\sigma)}{\frac{k^{(2)}}{\alpha_2}\OPT(\sigma)} = \frac{\alpha_2}{\lambda}.
\end{align*}

\paragraph{Case-II: Usage of $\PRM_{\boldsymbol{\phi_1}}$}
If $\max(\sigma) \geq \alpha_1 p_{\min}$, the competitive ratio improves over the cases before $ T_{\textsf{pred}} $:
\begin{align*}
    \gamma(\lambda) &= \frac{\OPT_1(\sigma) + \OPT_2(\sigma)}{\PRM_{\boldsymbol{\phi_1}}(\sigma|\alpha_1) + \PRM_{\boldsymbol{\phi_2}}(\sigma|\alpha_2)}\\
    &\leq \frac{\OPT(\sigma)}{\frac{k^{(1)}}{k \alpha_1} (\OPT(\sigma) - k \alpha_1 p_{\min}) + \frac{k^{(2)}}{k \alpha_2} \OPT(\sigma)}\\
    &= \frac{1}{\frac{1 - \lambda}{\alpha_1} \left(1 - k \alpha_1 \frac{p_{\min}}{\OPT(\sigma)}\right) + \frac{\lambda}{\alpha_2}},
\end{align*}
where the inequality follows from the results in Appendix \ref{proof_bound_cr}, without considering the forced trading phase (since now the trading stops before $ T_{\textsf{pred}} $): 
\begin{align*}
    \PRM_{\boldsymbol{\phi_1}}(\sigma|\alpha_1) \ge \frac{1}{\alpha_1} \OPT_1(\sigma) - k^{(1)} p_{\min}.
\end{align*}
When $\OPT(\sigma)$ increases, the bound on $\gamma(\lambda)$ decreases. The worst-case bound is achieved when $\OPT_1(\sigma) = k^{(1)} \alpha_1 p_{\min}$ and 
\begin{align*}
    \OPT(\sigma) = \frac{k}{k^{(1)}} \OPT_1(\sigma) =k \alpha_1 p_{\min}.
\end{align*}
Consequently, we have
\begin{align*}
    \gamma(\lambda) \leq \frac{1}{\frac{\lambda}{\alpha_2}} = \frac{\alpha_2}{\lambda}.
\end{align*}

Thus, the overall robustness is determined by taking the maximum of the two cases ($T > T_{\textsf{pred}}$ and $T \le T_{\textsf{pred}}$), leading to
\begin{align*}
    \gamma(\lambda) \leq \max \left\{ \frac{\alpha_2}{\lambda}, \frac{\theta}{1 - \lambda + \frac{\lambda}{\alpha_2} \theta}\right\} \leq \frac{\alpha_2}{\lambda}.
\end{align*}

To prove the last inequality, we need to show that
\begin{align*}
    \frac{\alpha_2}{\lambda} \ge \frac{\theta}{1 - \lambda + \frac{\lambda}{\alpha_2} \theta}
\end{align*}
always holds. Rearranging terms, we get
\begin{align*}
    \alpha_2 (1 - \lambda) + \lambda \theta \ge \lambda \theta,
\end{align*}
or equivalently,
\begin{align*}
    \alpha_2 (1 - \lambda) \ge 0.
\end{align*}
Since $\lambda \in [0,1] $ and $\alpha_2 > 1$, the above inequality is always satisfied, confirming that
\begin{align*}
    \max \left\{ \frac{\alpha_2}{\lambda}, \frac{\theta}{1 - \lambda + \frac{\lambda}{\alpha_2} \theta}\right\} = \frac{\alpha_2}{\lambda}.
\end{align*}
We thus complete the proof.        

\end{document}

%% file: commands.tex
\usepackage[lined, ruled]{algorithm2e} 
    
    \SetAlFnt{\small}
    \SetAlCapFnt{\small}
    \SetAlCapNameFnt{\small}
    \SetAlCapHSkip{0pt} 
    \IncMargin{-\parindent}

\usepackage{enumitem}
\setlist{nosep,topsep=0pt}

\usepackage{xspace,nicefrac,tcolorbox,xifthen,tikz,subcaption,enumitem,physics,suffix,wrapfig}
\usetikzlibrary{shapes, arrows, positioning}

\usepackage{xcolor}
    \definecolor{myred}{HTML}{ea4335}
    \definecolor{mygreen}{HTML}{41a756}
    \definecolor{myblue}{HTML}{4285f4}

\usepackage[capitalize]{cleveref}

\renewcommand{\paragraph}[1]{\smallskip\noindent\textbf{#1.}}

\DeclareMathOperator*{\argmax}{\arg\max}

\usepackage{xifthen}
\newcommand{\Line}[4]{%
    #1&%
    \;\ifthenelse{\isempty{#2}}{\phantom{=}}{#2}\;%
    #3%
    \ifthenelse{\isempty{#4}}{}{&&\qquad\left(#4\right)}%
}

\newcommand{\ECR}{\textsf{ECR}\xspace}

\newcommand{\OPT}{\textsf{OPT}\xspace}
\newcommand{\ALG}{\textsf{ALG}\xspace}
\newcommand{\OC}{\textsf{OC}\xspace}
\newcommand{\PRM}{\textsf{PseudoMax}\xspace}

\newcommand{\CRs}{\textsf{CR}s\xspace}
\newcommand{\CR}{\textsf{CR}\xspace}
\newcommand{\OCK}{\textsc{OC-Known}\xspace}
\newcommand{\OCN}{\textsc{OC-Notice}\xspace}
\newcommand{\OCU}{\textsc{OC-Unknown}\xspace}
\newcommand{\OCP}{\textsc{OC-Prediction}\xspace}